\documentclass[12pt]{article}
\usepackage{amsmath,amsfonts,amsthm,array}
\usepackage{commath}
\usepackage{bm,color,setspace,graphicx,times,subcaption}
\usepackage{mathtools}
\usepackage{hyperref}
\usepackage{algorithm}
\usepackage[noend]{algpseudocode}
\usepackage{comment}
\usepackage[margin=1in]{geometry}
\graphicspath{{Figs/}}

\theoremstyle{plain}
\newtheorem{thm}{Theorem}[section]
\newtheorem{lem}[thm]{Lemma}

\newtheorem{cor}[thm]{Corollary}

\usepackage[normalem]{ulem}

\definecolor{darkgrn}{rgb}{0, 0.8, 0}

\newcommand{\R}{ {\mathbb R} }

\newcommand{\tG}{\tilde{G}}

\begin{document}

\author{Prashant Gupta, Yiran Guo, Narasimha Boddeti, Bala Krishnamoorthy\\
Washington State University}

\date{}

\title{SFCDecomp: Multicriteria Optimized Tool Path Planning in 3D Printing using Space-Filling Curve Based Domain Decomposition}

\maketitle

\begin{abstract}
  We explore efficient optimization of toolpaths based on multiple criteria for large instances of 3D printing problems.
  We first show that the minimum turn cost 3D printing problem is NP-hard, even when the region is a simple polygon.  
  We develop \emph{SFCDecomp}, a space filling curve based decomposition framework to solve large instances of 3D printing problems efficiently by solving these optimization subproblems independently.
  For the Buddha model, our framework builds toolpaths over a total of 799,716 nodes across 169 layers, and for the Bunny model it builds toolpaths over 812,733 nodes across 360 layers.
  Building on SFCDecomp, we develop a multicriteria optimization approach for toolpath planning.
  We demonstrate the utility of our framework by maximizing or minimizing tool path edge overlap between adjacent layers, while jointly minimizing turn costs.
  Strength testing of a tensile test specimen printed with tool paths that maximize or minimize adjacent layer edge overlaps reveal significant differences in tensile strength between the two classes of prints.


\end{abstract}

\section{Introduction}
  We study {\it dense infill} 3D printing problems, where a given region is completely covered by depositing material with an extruder.
  Design of the tool path, i.e., the sequence in which the extruder moves while depositing material, has crucial implications on print quality as well as mechanical properties of the printed object.
  The extruder can go over non-print or previously printed regions with idle movements.
  Two problems closely related to 3D printing are milling and lawn mowing.
  But in the milling problem, the cutter cannot exit the region (pocket) that it has to cover.  
  The lawn mowing problem is similar to 3D printing problem since the cutter can mow over non grass as well as already mowed regions.
  But one wants to minimize non-print movement in 3D printing in order to improve efficiency.
  
  Various geometric tool path patterns are used such as zigzag, spiral, and contour parallel, but most of them suffer from directional bias.
  For instance, spiral and contour parallel tool paths do not allow cross weaving between adjacent layers.
  More generally, aspects of tool path design across multiple layers and their effects on mechanical properties of the printed objects have not been studied in detail.
  This motivated the development of our framework for optimization based tool path planning, where we can optimize the tool path based on multiple criteria.
  At the same time, we show that the 3D printing tool path optimization problem is NP-hard, and hence large instances become much harder to solve.
  One approach to handle large instances involves decomposition into subdomains, where the subproblems can be solved in parallel and the overall tool path designed by combining solutions for the subdomains.

\subsection{Our Contributions}
  We focus on the minimum turn, minimum edge cost, as well as combinations of these two 3D printing problems.  
  \begin{itemize}
    \item We show that minimum turn cost 3D printing problem is NP-hard, even when the region is a simple polygon.  
    \item We develop \emph{SFCDecomp}, a space filling curve based decomposition framework to solve large instances of 3D printing problems efficiently by solving subproblems from the decomposition independently.
      Our framework builds toolpaths over a total of 799,716 nodes across 169 layers for the Buddha model, and over 812,733 nodes across 360 layers for the Bunny model.
      See Figures \ref{fig:introimage} and \ref{fig:bunnybuddha} for sample layers and a print.

    \item Building on SFCDecomp, we develop a multicriteria optimization approach for toolpath planning.
      We demonstrate the utility of this approach by maximizing or minimizing tool path edge overlap between adjacent layers, while jointly minimizing turn costs.
    \item We measure the mechanical strength of prints with varying tool path edge overlaps across adjacent layers.
      Strength testing of tensile test specimens printed with tool paths that respectively maximize and minimize adjacent layer edge overlaps reveal significant differences in tensile strength between the two classes of prints.      
  \end{itemize}

\begin{figure}[ht!] 
  \centering
  \includegraphics[height=1.98in]{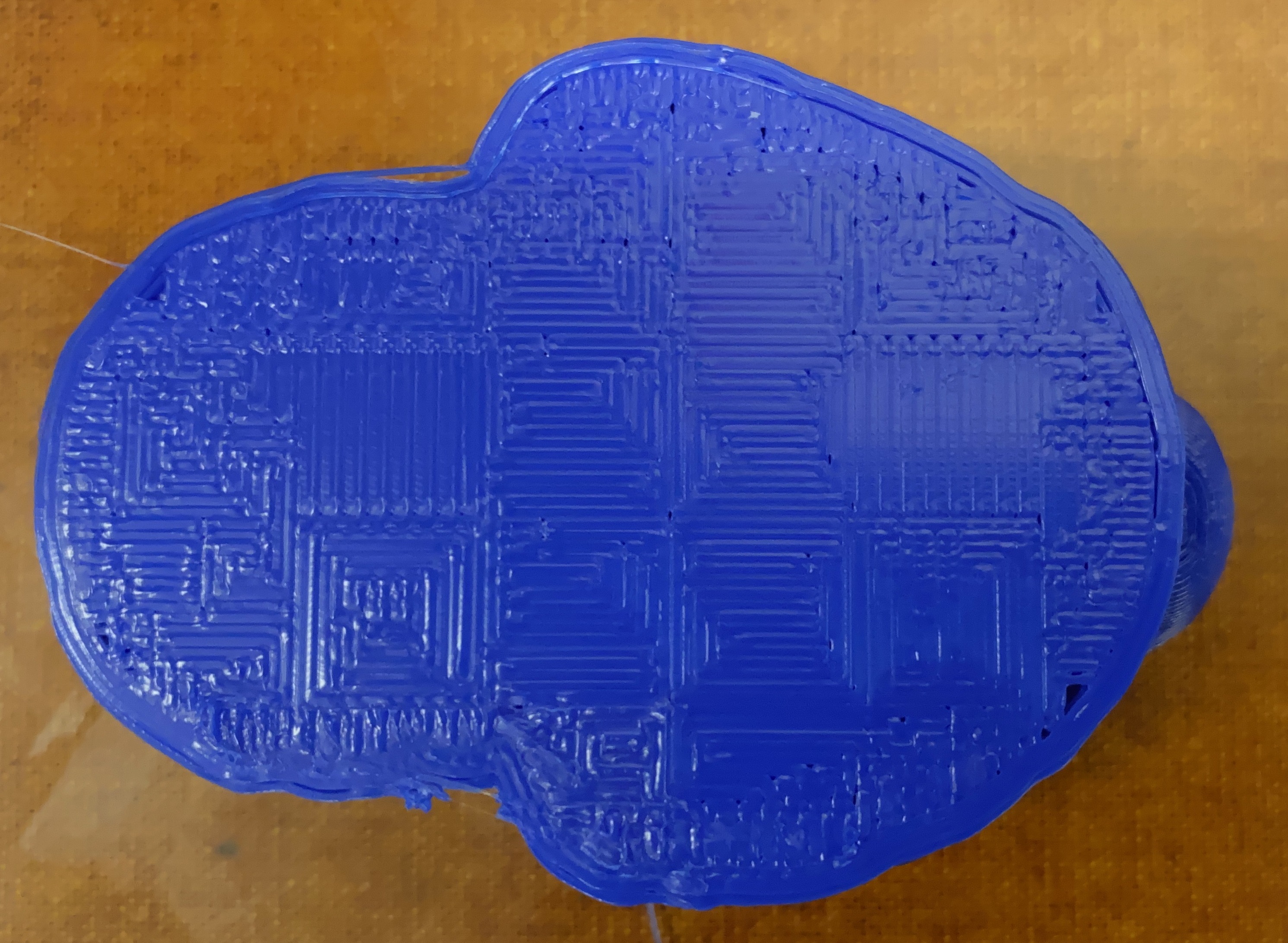}
  \includegraphics[height=1.98in]{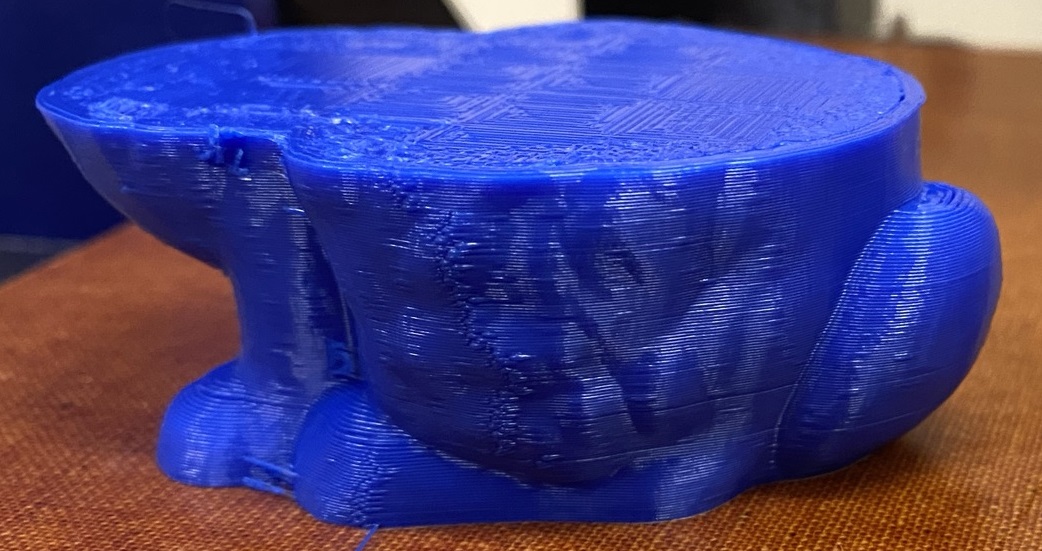}
  \caption{\label{fig:introimage} Top and side views of a print of the Bunny at Layer 124.}  	
\end{figure}
  
\subsection{Related Work} \label{sec:previouswork}

The lawn mowing, milling, and 3D printing problems are closely related to the more general geometric traveling salesman problem (TSP).

  \paragraph{Geometric TSP} \label{sec:geometricTSP}
  In the geometric traveling salesman problem (GTSP) with mobile clients, the objective is to find an optimal tour of the salesman to visit a given set of clients, each of whom can travel up to a distance $r$ to meet the salesman.
  Variants of the GTSP include the milling and lawn mowing problems. 
  Arkin et al.~\cite{ArFeMi2000} showed that the minimum length lawn mowing problem for polygonal regions with or without holes is NP-hard, and the minimum length milling problem for polygon regions with holes is also NP-hard.
  Arkin et al.~\cite{ArBeDeFeMiSe2005} also showed that the minimum turn milling problem is NP-hard.  
  
  The TSP with turn costs is known as the \emph{angular metric TSP}.
  Aggarwal et al.~\cite{AgCoKhMoSc2000} proved it is NP-hard. 
  Integer programming formulations of this problem are hard to solve to optimality \cite{AiAnFrJUlAl2017}. 
  Finding a tour connecting a given set of points such that angles between two adjacent edges of the tour is constrained was studied by Fekete and Woeginger \cite{FeWo1997}. 
  Reif and Wang \cite{ReWa1998} showed that the angle restricted shortest path problem with obstacles is NP-hard.
  In 3D printing, apart from computational challenges (NP-hardness), the optimal print direction can change based on mechanical factors such as temperature gradient to minimize thermal residual stress and strain.
  This aspect is a direct motivation for developing our framework.
  
  \paragraph{Tool Path Geometry}\label{sec:extruderpath}
  A popular tool path generation method uses zigzag patterns. 
  Space filling curves (SFCs) are also used in generating infill.
  An SFC in 2D is a continuous curve with positive Jordan content (i.e., area $>0$). 
  In 3D printing, the extruded bead has finite thickness and hence length of the SFC is finite. 
  Zhao et al.~\cite{ZhGuHuGaYoChBeZhCoDaBa2016} developed a Fermat Spiral infill (SFC) based smooth tool path optimized for continuity. 
  Contour parallel tool paths follow the boundary of the polygon \cite{YaLoFuWa2002}.
  Although both these methods give curved tool paths, the Fermat Spiral infill \cite{ZhGuHuGaYoChBeZhCoDaBa2016} is smoother.
  Curved tool paths in FDM are usually approximated by piecewise linear line segments.
  Hence highly curved tool paths require increasingly short line segments in the linear approximation.
  Zhao et al.~\cite{ZhGuHuGaYoChBeZhCoDaBa2016} discussed this limitation of their method for lower end 3D printers.
  They also pointed out that both spiral and contour parallel infills suffer from directional bias due to which adjacent layers cannot cross-weave at an angle.
  This gives zigzag tool paths some advantages over spirals and contour parallel tool paths since zigzag consists of linear segments and allows cross weaving between adjacent layers.

  Kuipers et al.~\cite{KuWuWa2019} developed \emph{CrossFill}, an approach based on SFCs to generate continuous paths for \emph{sparse} infill 3D printing.  
  Wasser et al.~\cite{WaJaPi1999} suggested fractal-like SFC for infill based on a TSP heuristic, but their method cannot handle turn costs and have ambiguity on the input graph for general polygons.
  Bertoldi et al.~\cite{BeYaPiGu1998} proposed a method that uses domain decomposition with a classical SFC (Hilbert curve) to find the tool path.
  But the Hilbert curve imposes restrictions on print directions.
  In our framework, we use a quadtree for domain decomposition and classical SFC (Hilbert curve) to create sequences of these domains.
  We then employ optimization that could be based on multiple criteria to find the tool path in each decomposed subdomain. 

  \paragraph{Tool Path Optimization}\label{sec:optimization}
  The tool path in each layer can go straight or take turns (within the plane) to fill the layer, and its overall shape affects various quality factors.
  One often tries to optimize the tool path in each layer for its continuity, smoothness, or both.
  But optimizing for multiple quality factors at one time could be highly inefficient.
  For instance, spiral and contour parallel infills try to optimize smoothness and continuity of the tool path, but cannot consider cross weaving between layers due to directional bias \cite{GiRoSt2015}.

  Bedel et al.~\cite{BeCoMaNiWhLe2021} have recently studied the optimization of space-filling curves under orientation objectives.
  The orientations are modeled by a vector field that is typically created for the whole print domain.
  In the first step, a Hamiltonian tour is identified for the whole domain by solving a combinatorial optimization problem.
  This cycle is then modified using local stochastic optimization steps to optimize for alignment with the vector field as well as smoothness and coverage while maintaining Hamiltonicity.
  On the other hand, the SFCDecomp framework identifies Hamiltonian cycles only for the individual subdomains identified by the decomposition and hence has the potential to scale better for larger models.
  Further, all computations for the subdomains can be handled independently or in parallel.

    Fang et al.~\cite{FaZhZhChZhWa2020} presented a computational framework to generate spatially designed layers for 3D printing that are not necessarily planar.
    The toolpaths on each layer are designed for multi-axis 3D printing systems and optimized by aligning filaments along maximum principal stress directions obtained from finite element analysis of the structure with prescribed mechanical loads.
    This framework also handles the entire print as a single domain.
    For certain samples, surfaces printed using this non-planar framework reported higher capacities to withstand loads compared to samples printed using planar layer-based FDM approach.
  
  In the setting of graphs for metric TSP, we can find the initial tour using, e.g., the Christofides--Serdyukov algorithm \cite{Ch1976,Se1978,vBeSl2020} and use a $k$-OPT heuristic \cite{ApBiChCo2007} to further improve the solution.
  The $k$-OPT heuristic solves a decision problem where we ask if a given tour can be improved by replacing $k$ edges in the tour with $k$ new edges.
  But the Christofides--Serdyukov algorithm might not give us accurate results since the problem is non-metric.
  With the simplest choice of $k=2$, the $2$-OPT has $O(n^2)$ computational complexity to find a local optimal solution compared to initial solution \cite{JoMc1997}.
  To the best of our knowledge, most heuristic algorithms for TSP have $O(n^2)$ complexity.
  We could partition the domain into multiple subgraphs and apply some heuristics on each subgraph \cite{Wa2004}.
  At the same time, we seek exact solutions for each subgraph since the problem is non-metric, so that the overall quality of the tool path is maximized.
  Lensgraf, Mettu, and coauthors \cite{LeMe2016,LeMe2017,LeMe2018,YoLeFiClMe2020} have used graph frameworks based on search algorithms, but do not include costs for change in direction.
  In contrast, our graph and cost-based optimization framework for complete infill problems can model new quality factors by varying the edge costs, and also models turn costs.
  In fact, our graph optimization framework can employ user-defined costs that capture various quality factors.
  
  \paragraph{Domain Decomposition}\label{sec:domaindecomposition}
  The geometric TSP is NP-hard \cite{Pa1977}.
  An alternative approach is to use domain decomposition, solve a TSP for the infill in each subdomain, and then connect the individual subdomain paths. 
  Chazelle and Palios \cite{ChPa1994} presented a review of some decomposition strategies.
  Convex decomposition of polygons is a well studied problem, e.g., see the book by Keil \cite{Ke2000}. 
  Exact convex decomposition for simple polygons without holes can be computed efficiently \cite{ChDo1979,ChDo1985}, but is NP-hard for polygonal regions with holes \cite{AnMoEr1982}. 
  An alternative can be approximate convex decomposition \cite{LiAm2006}.
  
  Domain decomposition has been used in 3D printing applications \cite{DiPaCuLi2014,DwKo2004,JiHeFuZhDu2017} to subdivide the polygon into sub-polygonal regions and find a cycle for each sub-polygon using closed zigzag curves to cover most of the vertices in them, and then join these cycles to find a complete tour.
  But their geometric decomposition does not guarantee existence of feasible dual graph of each sub-polygon, whereas our decomposition approach guarantees existence of dual pixel graph of each sub-polygon.
  We will primarily focus on finding paths and connecting them to get a complete path, but our work can be extended to finding complete tour by joining cycles.
  In computation, we guarantee the existence of an optimal connected path in each sub-polygon.
  We can vary the decomposition between alternate layers.
  The complete path generated by our method can have discontinuities \emph{only} at the boundaries of the polygon.
  
  To summarize, in the paragraph on \emph{Tool Path Geometry} we motivated the use of rectilinear tool paths, and that on \emph{Tool Path Optimization} we motivated the use of optimization tools.
  But TSP and variants are NP-hard in general, and that motivated our use of \emph{Domain Decomposition}.

\section{Preliminaries}
Let the extruder $\xi$ be the axis aligned unit square.
$\xi(p)$ denotes the placement of $\xi$ at point $p \in \R^2$ as its center.
The \textit{ {\bfseries Geometric 3D printing problem} (3dPP) on a polygon $R$ is to find a path/tour $\pi$ such that every point in $R$ is covered by the placement of $\xi$ on $\pi$, i.e., $R \subseteq \cup_{p \in \pi} \xi(p)$, subject to total idle movement less than or equal to a positive constant $\epsilon$.}
Note that $\xi(p)$ can hit outside the region $R$ for some $p \in \pi$.
We consider minimizing total length (sum of edge weights, more generally) or total turn cost, or a combination of both. 
We restrict our attention to integral orthogonal polygonal regions with or without holes, and the extruder is taken as a unit square restricted to axis parallel motion.
All boundary turns are $90^{\circ}$ in an integral orthogonal region, and boundary vertices have integer coordinates.
It can be considered a union of pixels, i.e., unit squares with axis-parallel edges and integer vertices (see Figure \ref{fig:inteortho3dprintprob}).
Hence we name it the \emph{integral orthogonal $3d$-printing problem} (IO3dPP), and will also refer to it in short as 3dPP.
We prove that 3dPP is NP-hard. 

\begin{figure}[ht!]
  \bigskip
  \centering
  \includegraphics[scale=0.35]{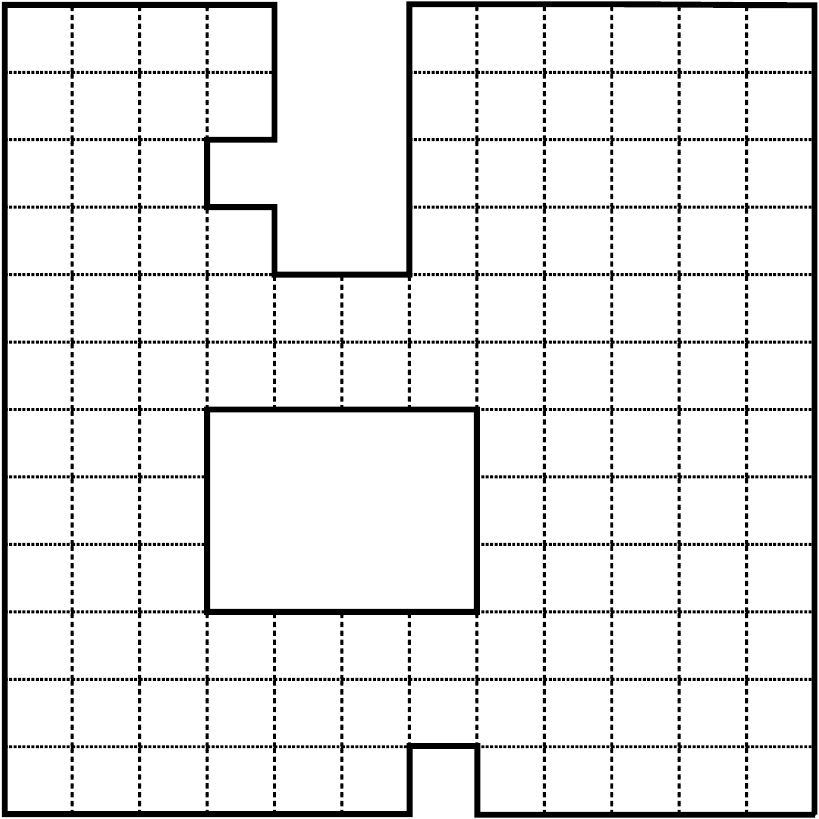}
  \hspace*{0.03in}
  \includegraphics[scale=0.35]{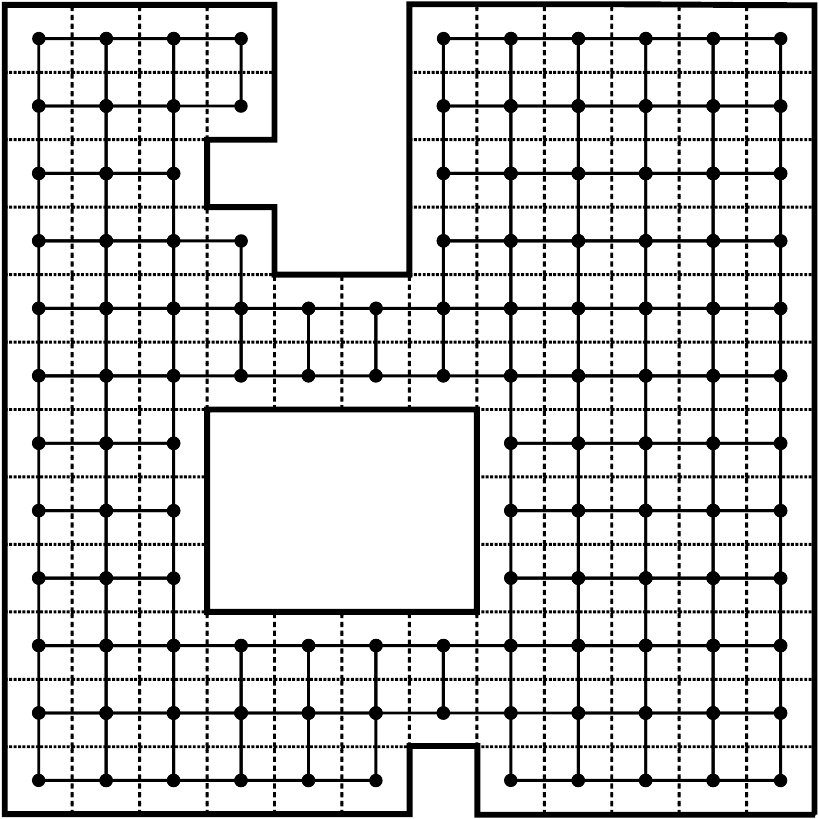}
  \caption{\label{fig:inteortho3dprintprob}
    An integral orthogonal region with a hole (left) and its dual graph (right).
  }
\end{figure}

We consider a connected, undirected, planar graph $G=(V,E)$ with $V= \{v_1, \dots, v_n\}$ and $E= \{(v_i, v_j): v_i, v_j \in V \text{ and } i \neq j\}$.
For each edge $(v_i, v_j) \in E$ there is a cost (or weight) $c_{ij}$ that depends on the Euclidean distance between vertices (if $(v_i, v_j) \not\in E$ then $c_{ij} = M$, some large positive value).
We take the \textit{turn cost} at vertex $v_i$ as $c_i = 0$ if the toolpath goes straight through $v_i$, and $c_i = 1$ if the path makes a turn at the vertex ($90^{\circ}$). 
For the default Euclidean minimum length problem, $c_{ij}$'s form a metric.
But $c_i$'s for the minimum turn problem do not form a metric \cite{ArBeDeFeMiSe2005}.

\section{NP-hardness proofs}

We show that NP-hardness of minimum length 3dPP follows directly from the known result on NP-hardness of the lawn mowing problem.
For the minimum turn cost 3dPP, we employ a two-step reduction from the problem of Hamiltonicity of square grid graphs to prove NP-hardness.

\subsection{Minimum Length 3dPP}
Arkin et al.~\cite{ArFeMi2000} showed that minimum length lawn moving problem for simple polygon or polygon with holes is NP-hard based on reduction from Hamiltonian circuit in planar bipartite graph with maximum degree 3 to Hamiltonian circuit in grid graphs.

\begin{lem} 
  Minimum length 3dPP is NP-hard for any connected polygon $R$ (with or without holes) and axis-aligned unit square extruder. 	
\end{lem}
\begin{proof}
  Proof of Arkin et al.~for lawn moving problem~\cite[Theorem 1]{ArFeMi2000} can be directly adapted to that for 3dPP where there is no idle movement for connected polygon $R$ with or without holes.
  We note that no point in $R$ is mowed more than once in their proof, and the total idle movement bound can be set as $\epsilon=1$.    
\end{proof}

\subsection{Minimum Turn Cost 3dPP}

We use a reduction similar to one introduced by Arkin et al.~\cite{ArBeDeFeMiSe2005}.
Previously, Itai et al.~\cite{ItPaSz1982} showed that the Hamiltonian circuit problem in grid graphs is NP-complete.
We first show that Hamiltonicity of square grid graph is reducible to that of Hamiltonicity of unit segment perpendicular end point intersection graph (HUSPEPIG) of axis aligned unit segments.
Unit segment perpendicular end point intersection graphs consist of unit horizontal or vertical segments that intersect only at end points (see Figure \ref{fig:USPIGNC_graph}).

\begin{figure}[hb!]
  \bigskip
  \centering
  \begin{subfigure}[t]{1.8in}
    \centering
    \includegraphics[scale=0.35]{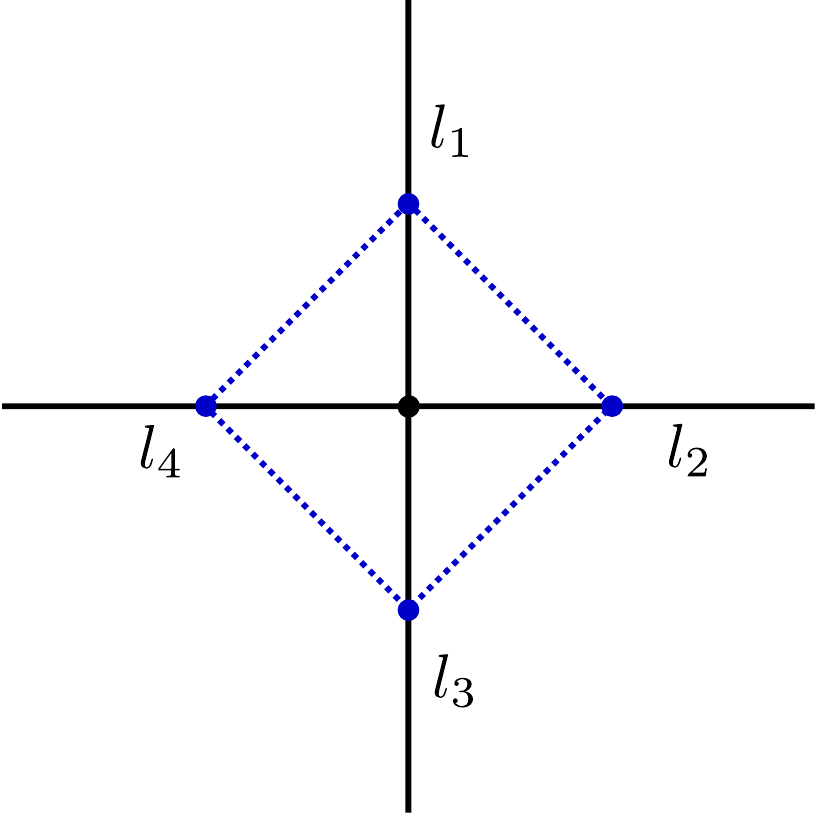}
    \caption{\label{fig:USPIGNC_grapha}}
  \end{subfigure}
  \quad
  \begin{subfigure}[t]{1.5in}
    \centering
    \includegraphics[scale=0.35]{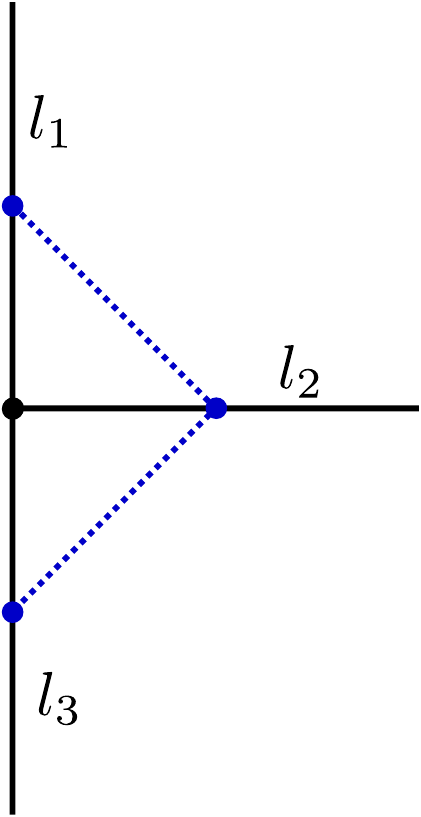}
    \caption{\label{fig:USPIGNC_graphb}}
  \end{subfigure}	
  \begin{subfigure}[t]{1.4in}
    \centering
    \vspace*{-1.4in}
    \includegraphics[scale=0.35]{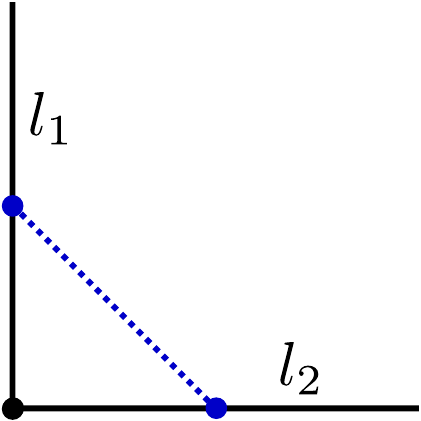}
    \vspace*{0.42in}
    \caption{\label{fig:USPIGNC_graphc}}
  \end{subfigure}
  \begin{subfigure}[t]{1.1in}
    \centering
    \includegraphics[scale=0.35]{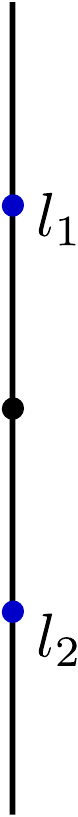}
    \caption{\label{fig:USPIGNC_graphd}}
  \end{subfigure}
  \caption{\label{fig:USPIGNC_graph} $l_1, l_2, l_3, l_4$ (black) are unit line segments.
    Figures (a), (b), (c), (d) show all possible intersections at the end points.
    Intersection graph (dotted blue lines) of corresponding  intersection of line segments are shown in Figures (a), (b), and (c).
    In Figure (d), unit segments intersect at an end point but are not perpendicular, so no edge is shown in the intersection graph.}
\end{figure}

\begin{lem} 
  Hamiltonicity of square grid graph is reducible to HUSPEPIG.	
\end{lem}
\begin{proof}
  Consider a bipartite grid graph $G$ with vertices having integer coordinates.
  Then $G$ can be represented as $2$-color graph.
  Rotate $G$ by $45^{\circ}$ and scale down edges in $G$ by $\sqrt{2}$.
  The length of each edge in this arrangement is $1/2$, the coordinates of each vertex are integer multiples of $1/2$, and the smallest distance between vertices of same color is $1$.
  Assign each white vertex a horizontal unit line segment and each black vertex a vertical unit line segment centered at the vertex to obtain the instance of HUSPEPIG (see Figure \ref{fig:Grid_to_USPIGNC}).
\end{proof}

\begin{figure}[ht!]
  \centering
  \includegraphics[scale=0.3]{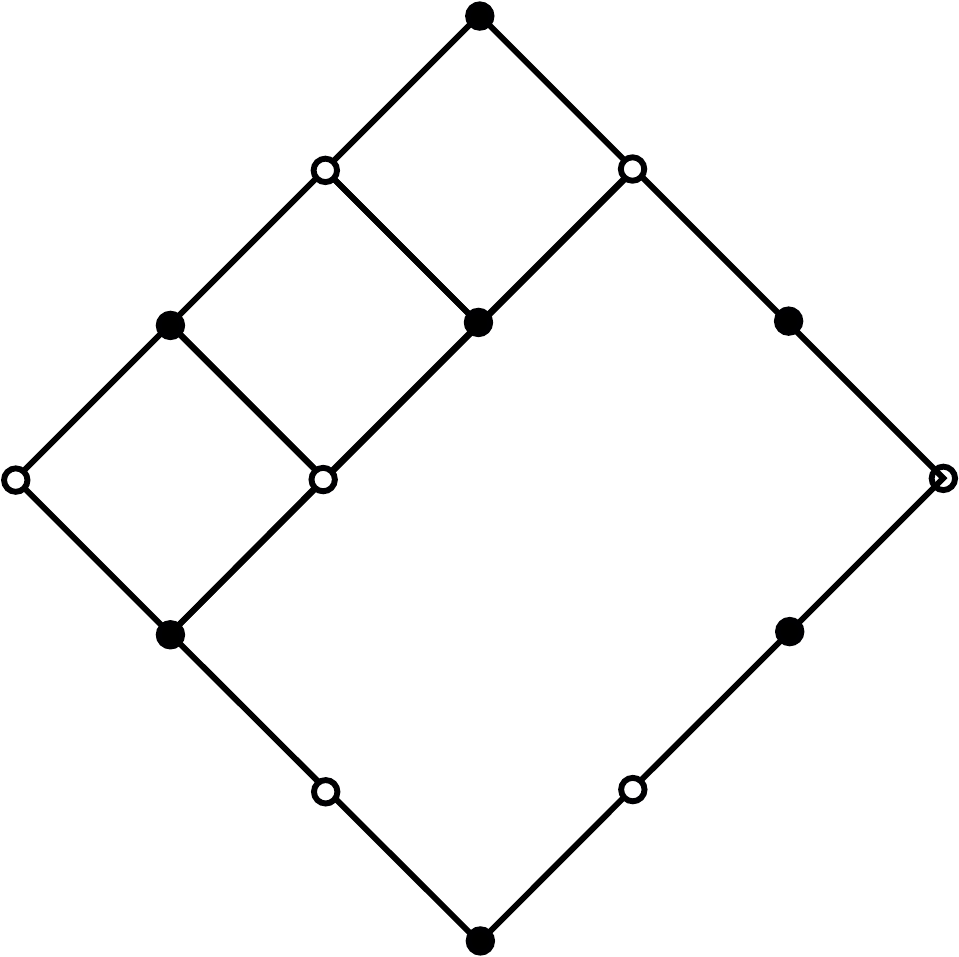}
  \quad\quad
  \includegraphics[scale=0.3]{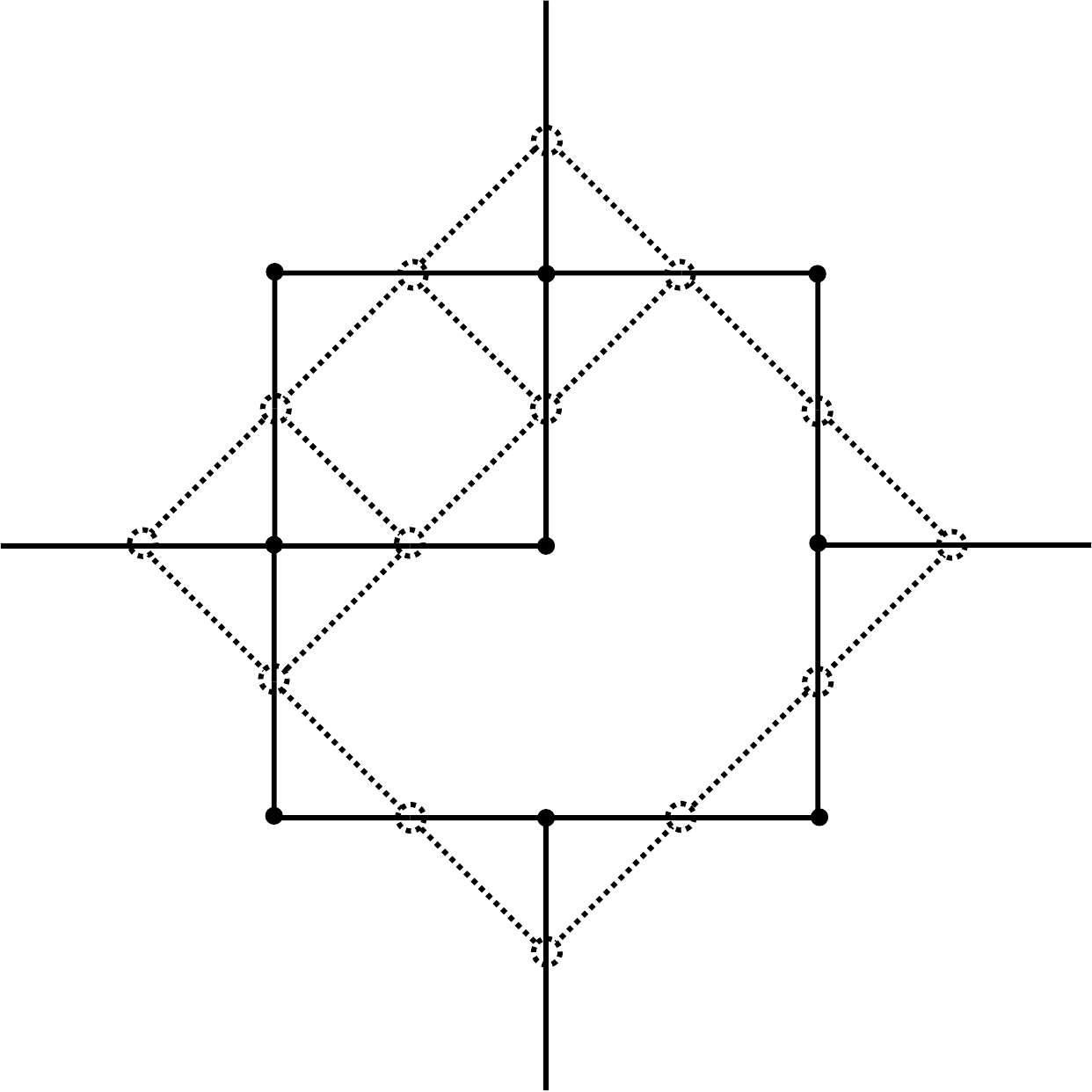}
  \caption{\label{fig:Grid_to_USPIGNC} Square grid graph $G$ (left).
    Dotted black graph (right) is the intersection graph of axis aligned unit line segments (black) intersecting at end points.
  }
\end{figure}

We now show that HUSPEPIG can be reduced to minimum turn cost 3dPP. 
Let each unit line segment be represented by a square block (Figure \ref{fig:Segment_to_Blocka}), which is the union of $9$ unit squares each representing the extruder.
Let the $9$-cluster $C$ be the dual graph of such a square block, consisting of $9$ vertices. 
With unit line segments represented by square blocks, there are $3$ types of intersection (Figures \ref{fig:Segment_to_Blockb}, \ref{fig:Segment_to_Blockc}, and \ref{fig:Segment_to_Blockd}).
Each $C$ has four corner vertices.
We can clearly find a Hamiltonian path that starts and ends at distinct corner vertices in $C$ (Figure \ref{fig:Turn_Cost}).
If the start and end vertices are on the same side of $C$, then the turn cost is $5$, else it is $4$.
We refer to these two traversals of $C$ as type-$1$ and type-$2$, and incur additional turn costs of $1$ and $0$ for entering and exiting $C$.

Figure \ref{fig:proofoutline} shows the outline of the argument.
We start with the set of axis parallel unit segments (Figure \ref{fig:proofoutlinea}), and its corresponding perpendicular intersection graph $G$ (Figure \ref{fig:proofoutlineb}) with its Hamiltonian cycle (in bold).
Assume without loss of generality that $G$ is connected and has $n>1$ vertices.
We replace unit line segments with square blocks (Figure \ref{fig:proofoutlinec}), which provides the connected polygonal region $R$.
Figure \ref{fig:proofoutlined} shows square blocks replaced by corresponding clusters, and Figure \ref{fig:proofoutlinee} shows the 3D printing tour to cover $R$.

\begin{figure}[ht!] 
  \centering
  \begin{subfigure}[t]{0.65in}
    \vspace*{-0.72in}
    \includegraphics[scale=0.22]{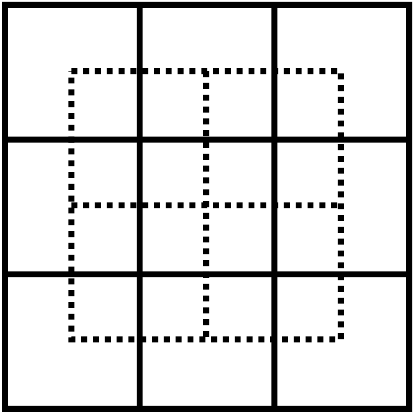}
    \vspace*{0.1in}
    \caption{\label{fig:Segment_to_Blocka}}
  \end{subfigure}
  \quad
  \begin{subfigure}[t]{1.9in}
    \centering
    \includegraphics[scale=0.22]{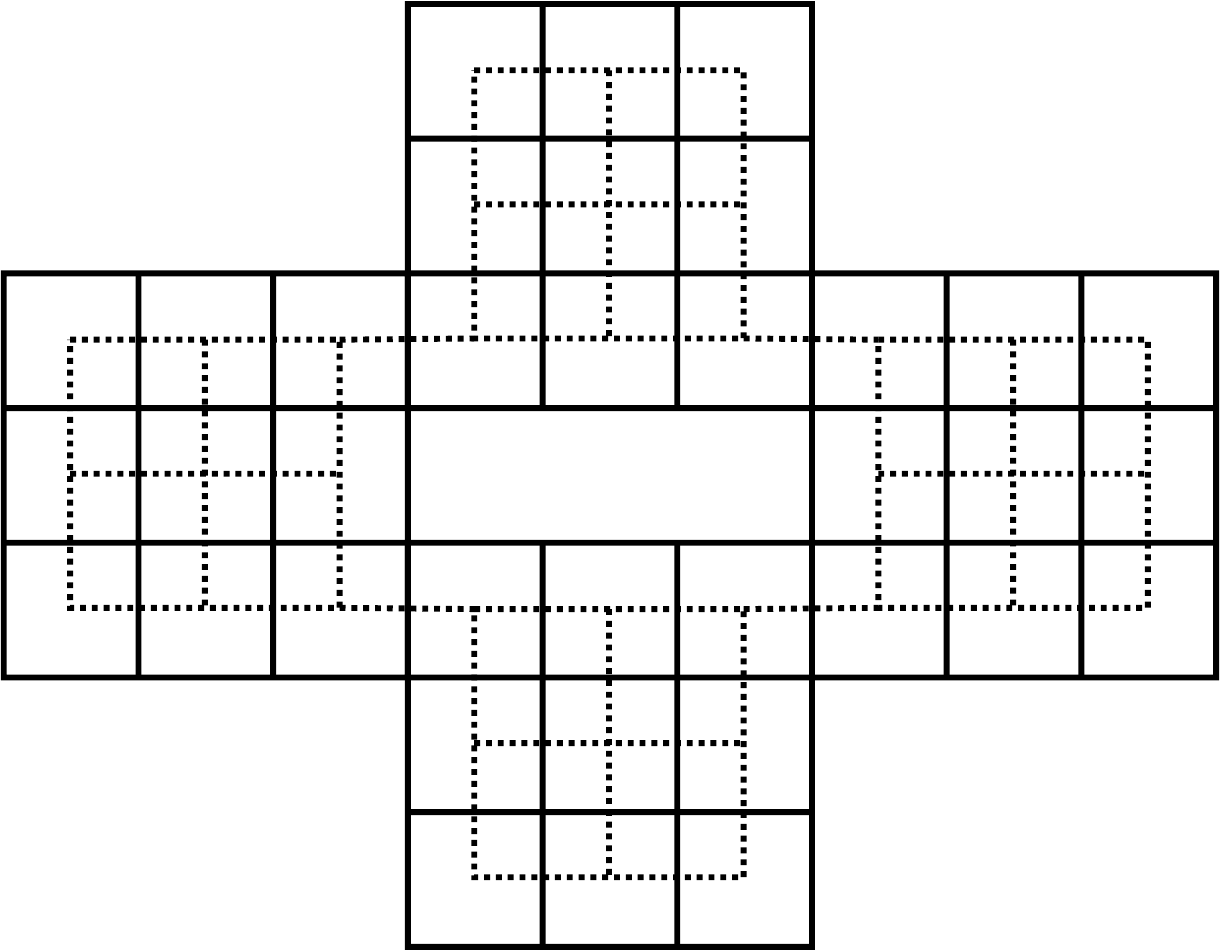}
    \caption{\label{fig:Segment_to_Blockb}}
  \end{subfigure}	
  \quad
  \begin{subfigure}[t]{1.9in}
    \centering
    \includegraphics[scale=0.22]{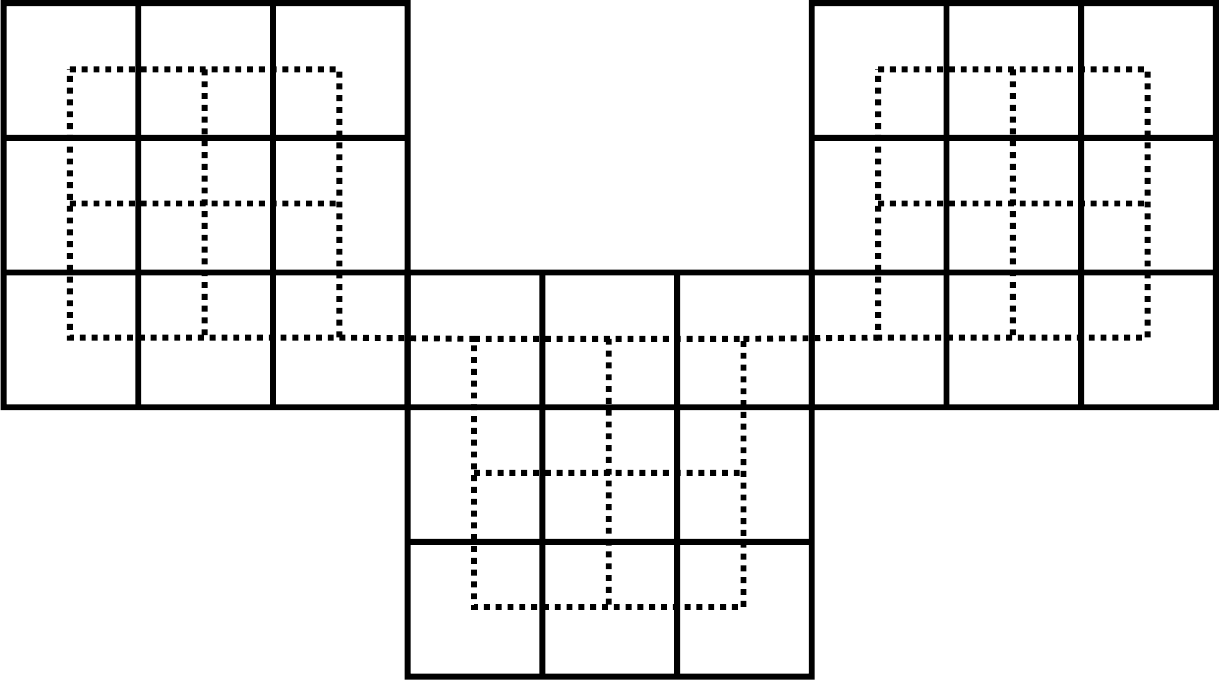}
    \caption{\label{fig:Segment_to_Blockc}}
  \end{subfigure}
  \quad
  \begin{subfigure}[t]{1.4in}
    \centering
    \includegraphics[scale=0.22]{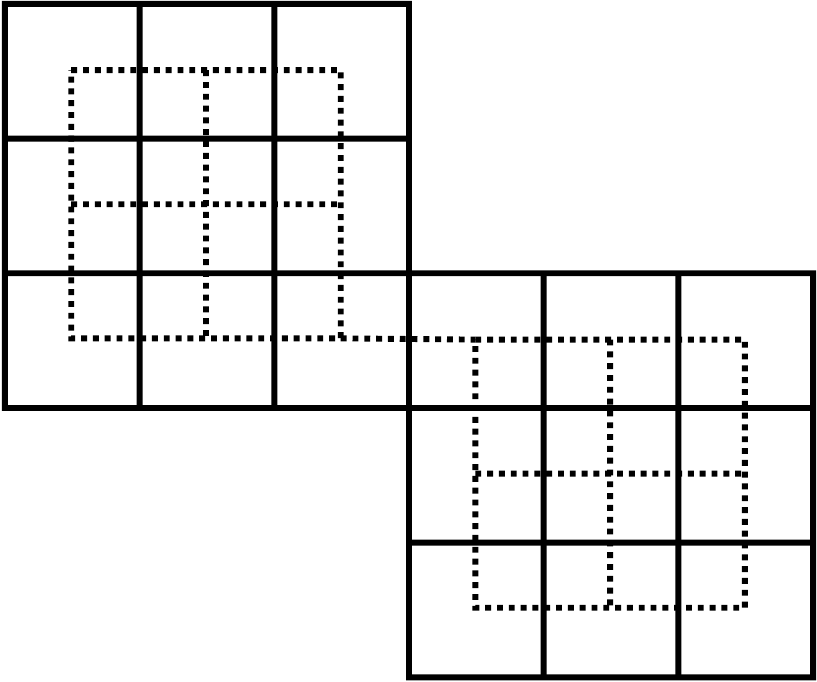}
    \caption{\label{fig:Segment_to_Blockd}}
  \end{subfigure}
  \caption{\label{fig:Segment_to_Block}
    (a): Each unit line segment is represented by a Square Block (solid black) and dual of the Square Block, i.e., its $9$-cluster (dotted black).
    (b), (c), (d): Connectivity based on corresponding unit line segment intersection shown in Figure \ref{fig:USPIGNC_grapha}, \ref{fig:USPIGNC_graphb}, and \ref{fig:USPIGNC_graphc}.
  }
\end{figure}

\begin{figure}[htp!] 
  \centering
  \includegraphics[scale=0.30]{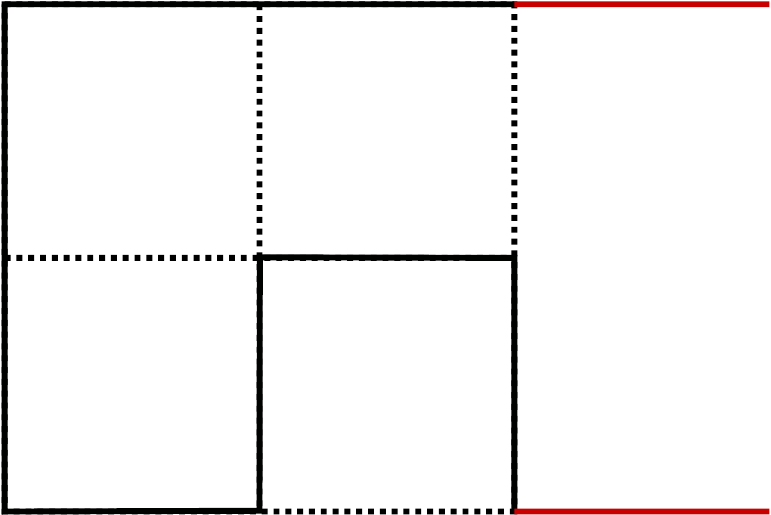}
  \quad\quad
  \includegraphics[scale=0.30]{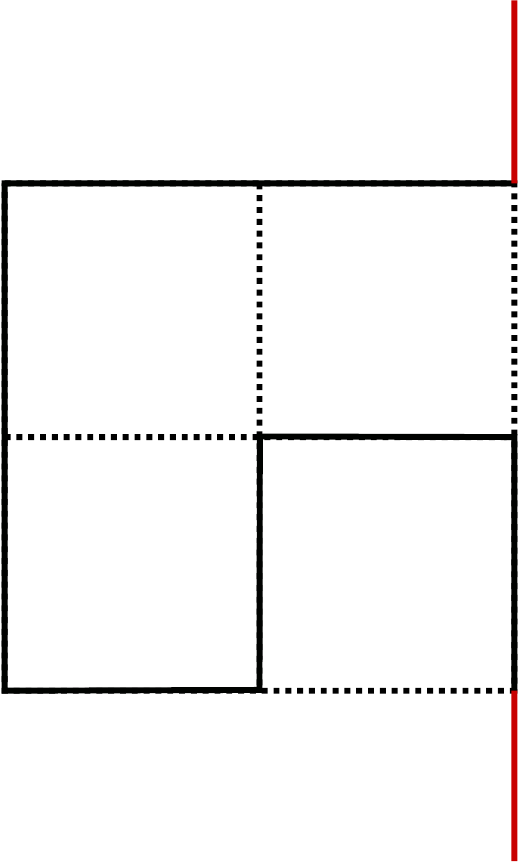}
  \quad\quad
  \includegraphics[scale=0.30]{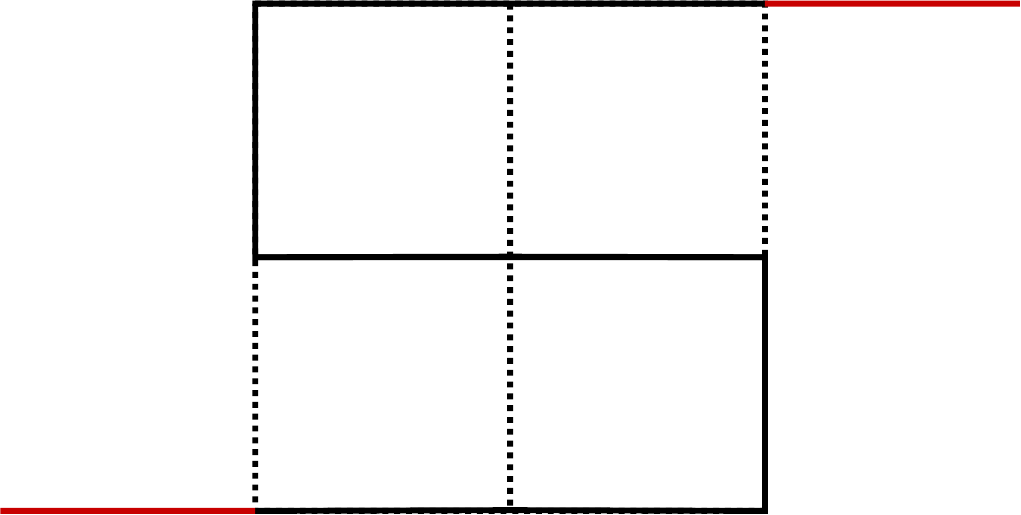}
  \caption{\label{fig:Turn_Cost}
    Left and Middle: Hamiltonian path (solid black) with start and end vertices on same side of $C_9$ with turns cost $5$.
    Total entry and exit (red) turn cost for $C_9$ is $1$.
    Right: Hamiltonian path (solid black) with start and end vertices at diagonally opposite corners of $C_9$, with turn cost $4$.
    Total entry and exit (red) turn cost for $C_9$ is $0$.
  }
\end{figure}

\begin{lem}\label{lem:sqblfirt}
  Any Hamiltonian tour of $R$ covers every $9$-cluster $C$ by a single path within $C$.  
\end{lem}

\begin{proof}
  A Hamiltonian tour $H$ of $R$ enters and exits any $9$-cluster $C$ through corner vertices.
  We have two distinct pairs of entry--exit vertices in $C$ for $H$ (as we have four corner vertices).
  Hence we can have at most two paths $p, p'$ within $C$ that are part of $H$.
  Further, since $H$ is a Hamiltonian tour, $p \cap p' = \emptyset$. 

  We consider two possible ways in which $p$ enters and exits $C$.
  First, let $p$ enter and exit at diagonal corner vertices of $C$, e.g., at $a, d$ partially covering $C$ (Figure \ref{fig:lemma33}).
  Since $p$ does not contain $b$ and $c$, and since $p \cap p' = \emptyset$, $p'$ cannot cover all remaining vertices of $C$ without intersecting $p$.
  One such case is shown in Figure \ref{fig:lemma33a}.
  This contradicts the Hamiltonicity of $H$.
  A similar argument holds when $p$ used $b$ and $c$ as end points. 

  Second, let $p$ enter and exit along the same side of $C$, e.g., using $a$ and $b$ (Figure \ref{fig:lemma33b}).
  Then $p'$ with end points $c, d$ cannot cover rest of the vertices of $C$ without intersecting $p$, raising a contradiction.
\end{proof}

\begin{figure}[htp!] 
  \centering
  \begin{subfigure}[t]{2.5in}
    \centering
    \includegraphics[scale=0.32]{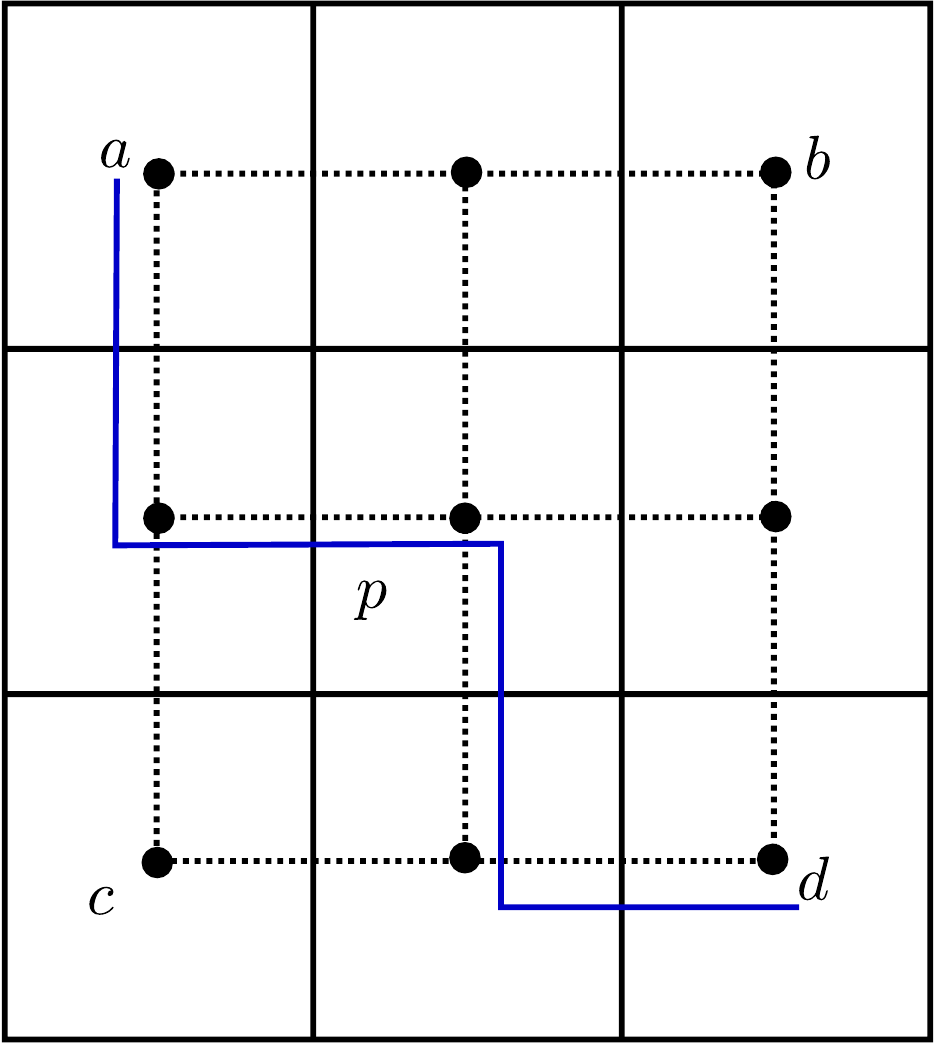}
    \caption{\label{fig:lemma33a}}
  \end{subfigure}
  \begin{subfigure}[t]{2.5in}
    \centering
    \includegraphics[scale=0.32]{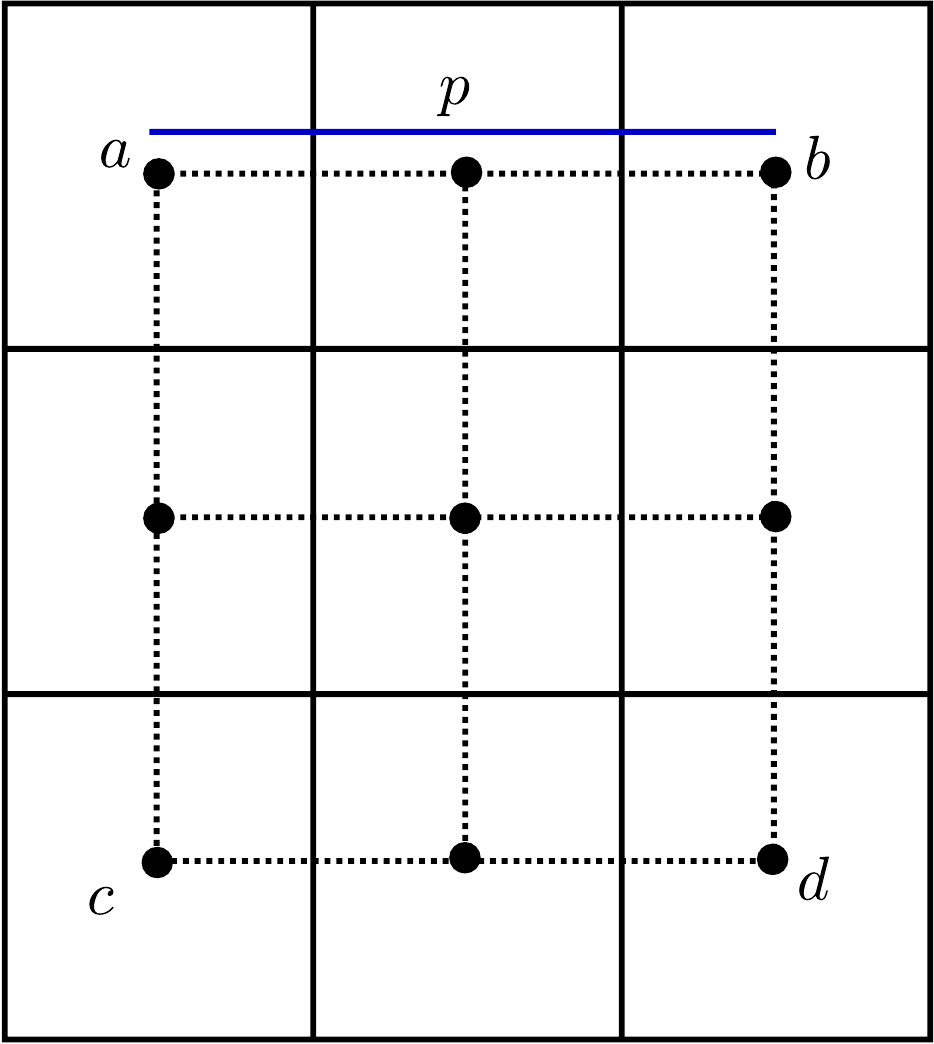}
    \caption{\label{fig:lemma33b}}
  \end{subfigure}	
  \caption{\label{fig:lemma33}
    Proof of Lemma \ref{lem:sqblfirt}.
  }
\end{figure}

\begin{thm}\label{thm:mtpnphard} 
  Minimum Turn 3dPP is NP-hard for any connected polygon $R$ with holes for axis-aligned unit square extruder. 	
\end{thm}
\begin{proof}
  Assume there exists a Hamiltonian tour $H$ in $G$.
  Any move in $H$ from vertex $v_i$ to $v_j$ is equivalent to moving from $9$-cluster $C_i$ to $C_j$ in $R$ by construction.
  By Lemma \ref{lem:sqblfirt}, $H$ covers each $9$-cluster $C_i$ by a single path within $C_i$.
  Further, $H$ uniquely determines the type of traversal (type $1$ or $2$) for each $9$-cluster $C_i$, and hence the numbers $t_1, t_2$ of these clusters such that $t_1 + t_2 = n$.
  This gives a 3dPP tour of $R$  with total turn cost $6t_1+4t_2$ (Figure \ref{fig:Turn_Cost}).

  Conversely, assume there is a 3dPP tour $T$ with turn cost $6t_1+4t_2$ for $t_1, t_2$ being the numbers of type-$1$ and type-$2$ traversals, and $t_1 + t_2 = n$.
  Since $T$ enters and exits each $9$-cluster exactly once (Lemma \ref{lem:sqblfirt}), we are guaranteed a tour of length $n$ in $G$ where each node is traversed exactly once.
  Thus $G$ has a Hamiltonian tour. 	
\end{proof}

\begin{figure}[ht!] 
  \begin{subfigure}[t]{3.22in}
    \centering
    \includegraphics[scale=0.387]{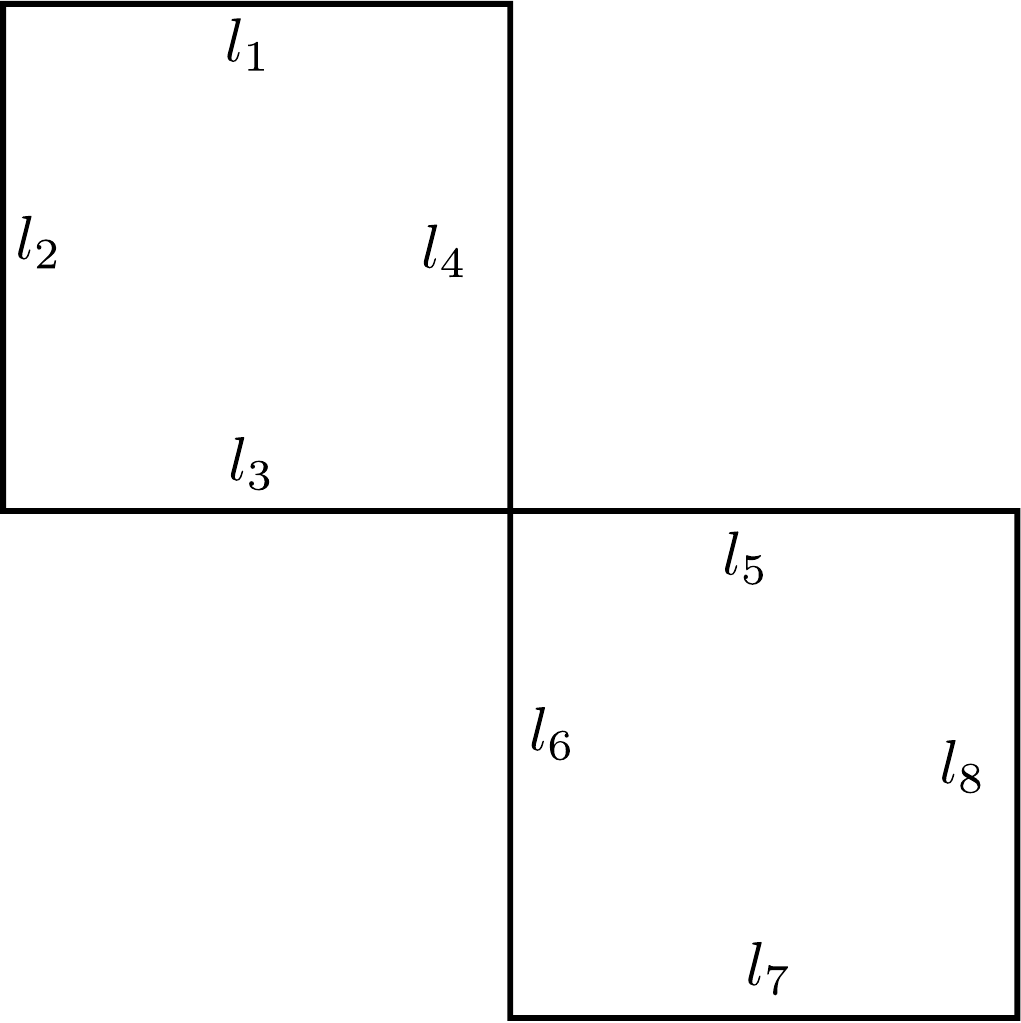}
    \caption{\label{fig:proofoutlinea}}
  \end{subfigure}
  \begin{subfigure}[t]{3.22in}
    \centering
    \includegraphics[scale=0.387]{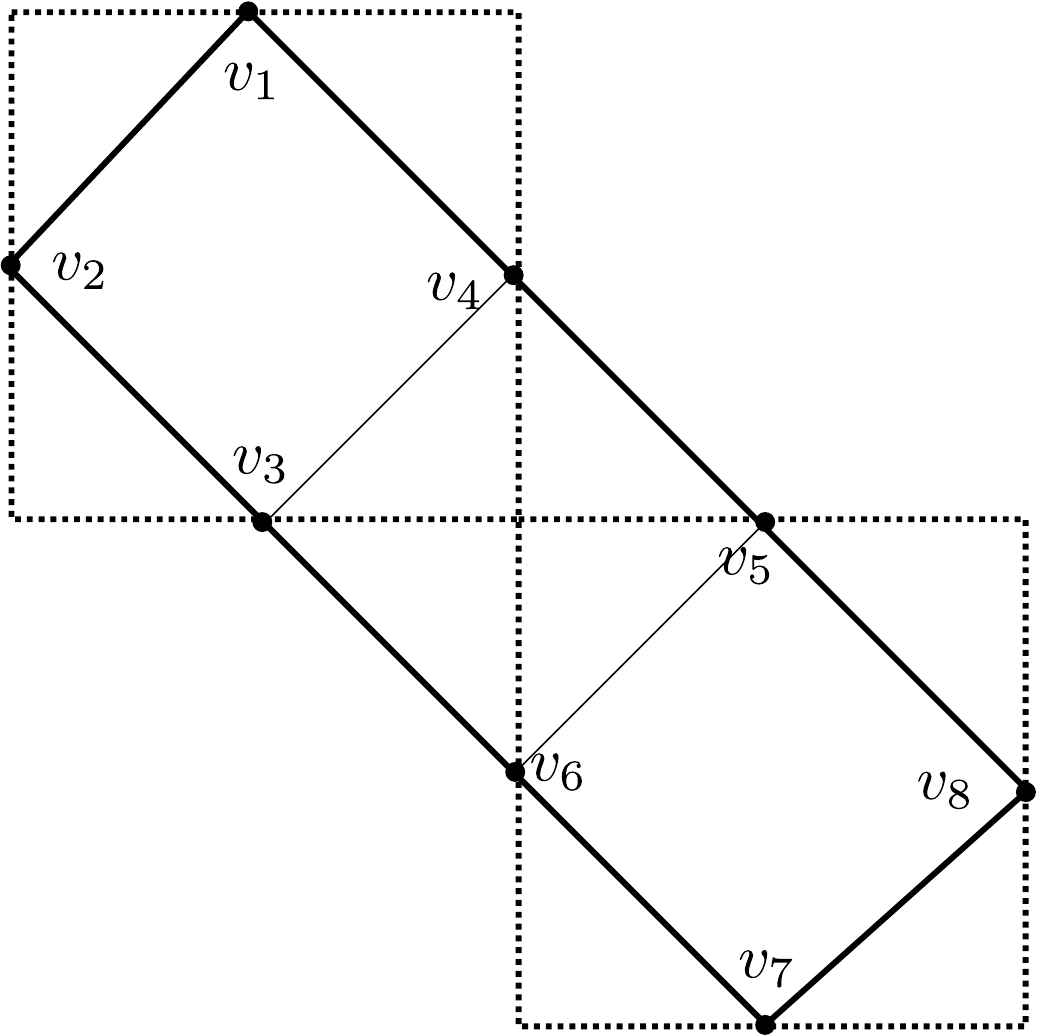}
    \caption{\label{fig:proofoutlineb}}
  \end{subfigure}
  \\
  \smallskip\\
  \begin{subfigure}[t]{3.22in}
    \centering
    \includegraphics[scale=0.387]{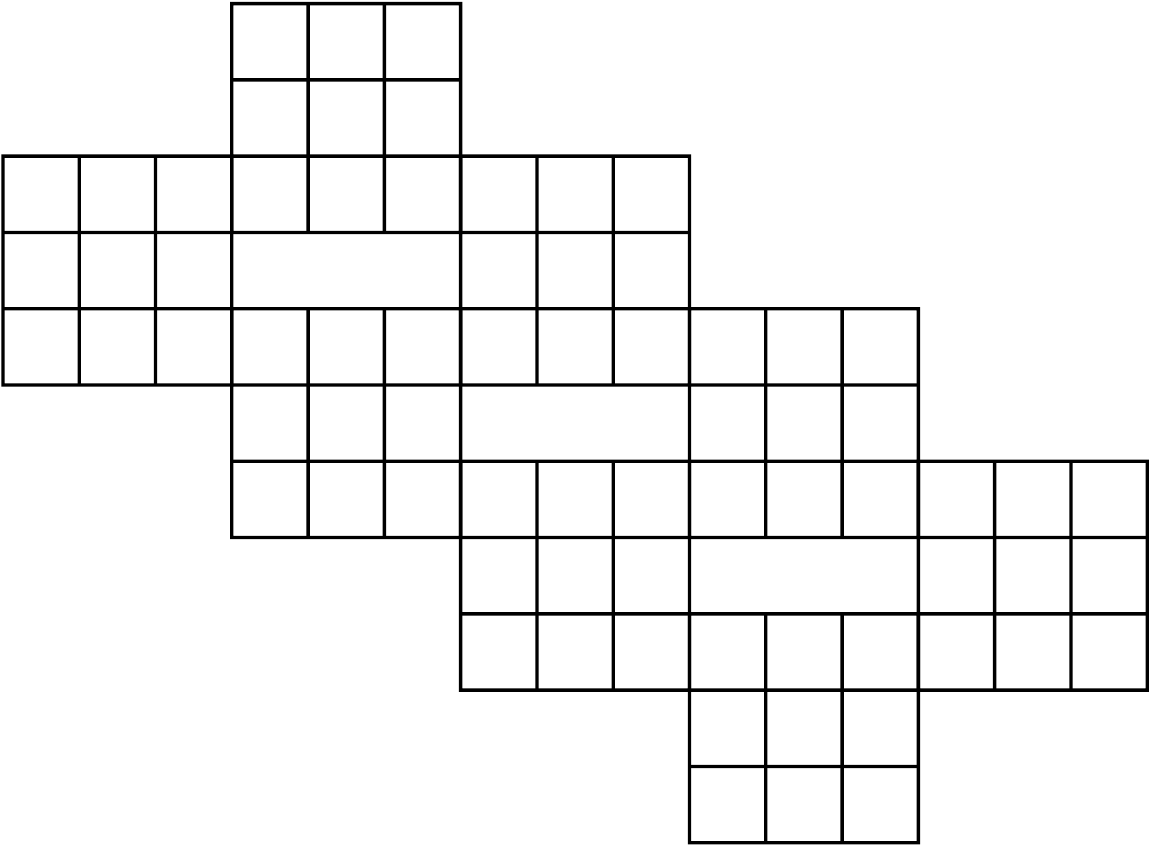}
    \caption{\label{fig:proofoutlinec}}
  \end{subfigure}
  \begin{subfigure}[t]{3.22in}
    \centering
    \includegraphics[scale=0.387]{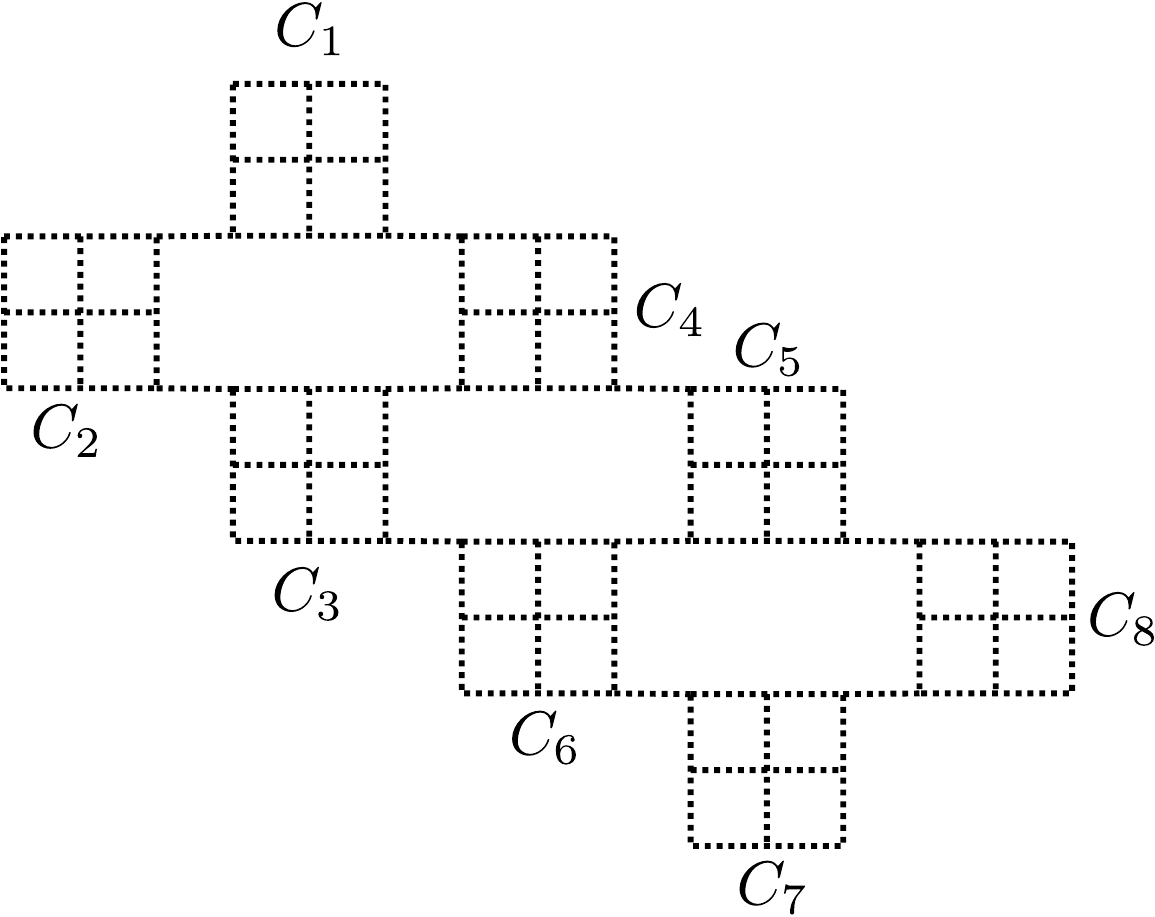}
    \caption{\label{fig:proofoutlined}}
  \end{subfigure}
  \\
  \smallskip  \\
  \begin{subfigure}[t]{3.22in}
    \centering
    \includegraphics[scale=0.387]{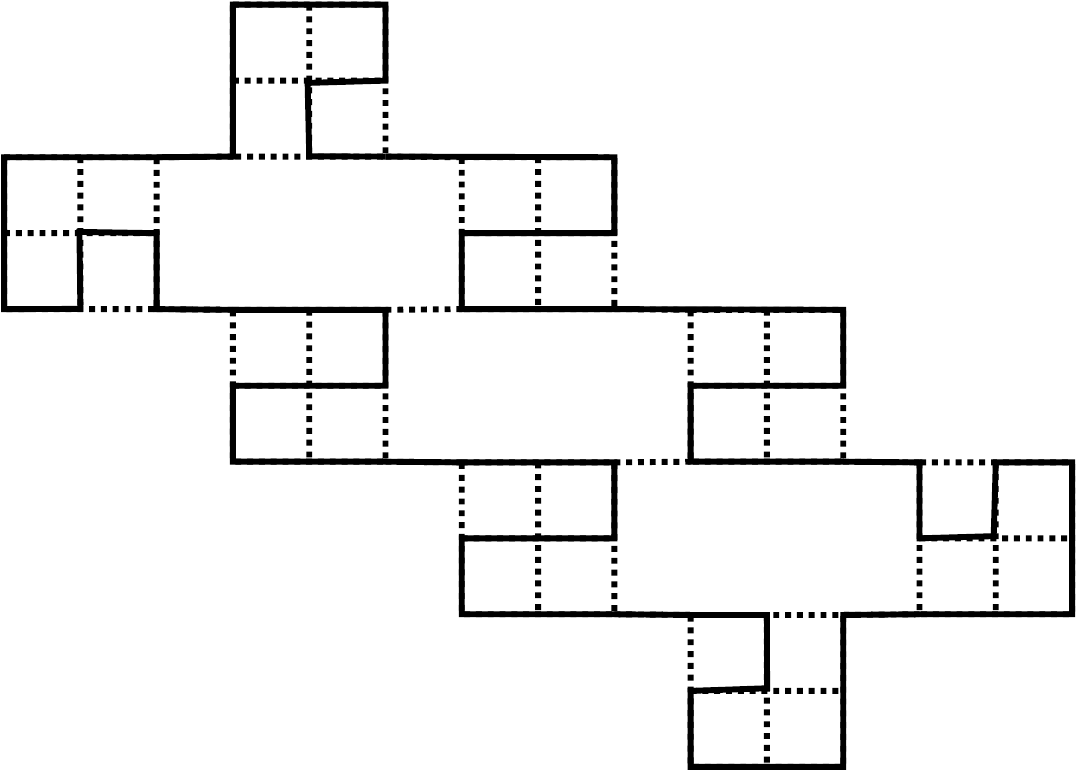}
    \caption{\label{fig:proofoutlinee}}
  \end{subfigure}
  \begin{subfigure}[t]{3.22in}
    \centering
    \includegraphics[scale=0.387]{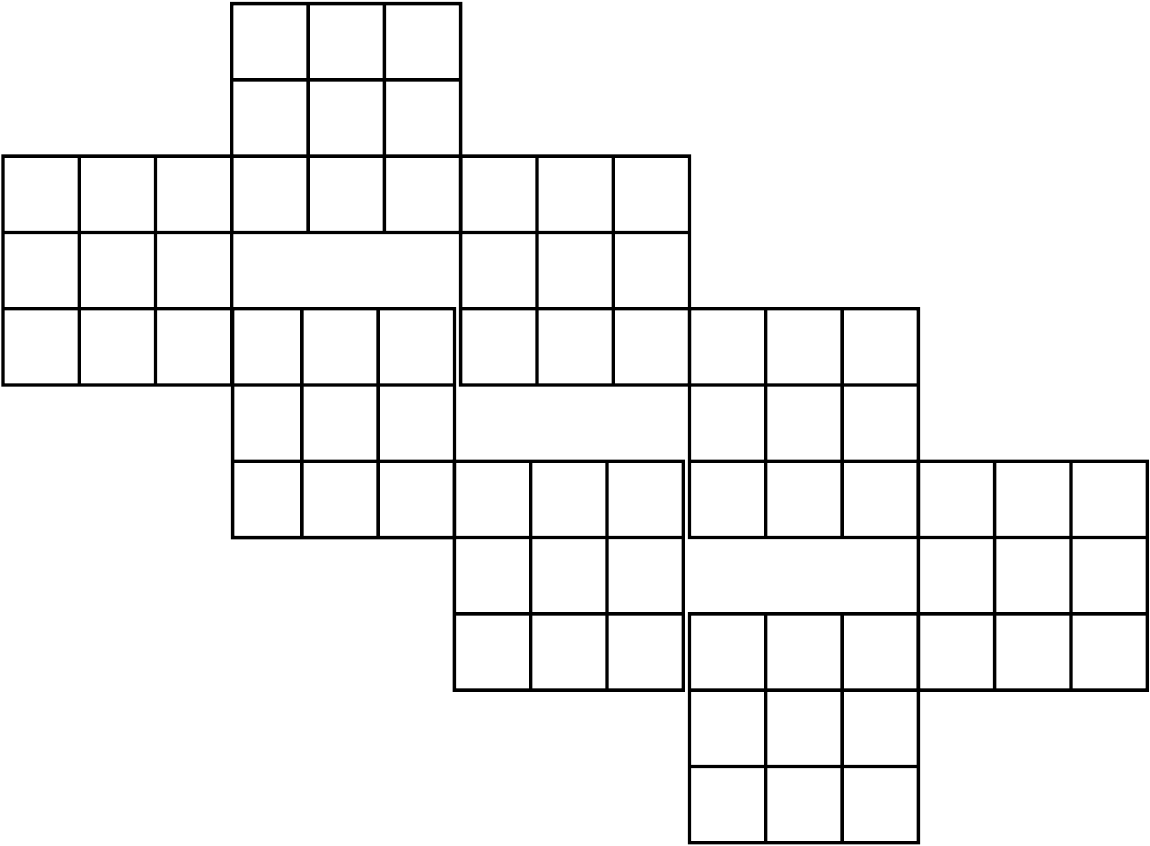}
    \caption{\label{fig:Square_Block_First_Simply_Connected}}
  \end{subfigure}
  \caption{\label{fig:proofoutline}
    See text above Lemma \ref{lem:sqblfirt} for details on Figures (\ref{fig:proofoutlinea})--(\ref{fig:proofoutlinee}).
    Proof of Corollary \ref{cor:nphardsplpoly} is illustrated in Figure (\ref{fig:Square_Block_First_Simply_Connected}).
  }
\end{figure}

\begin{cor} \label{cor:nphardsplpoly}
  Minimum Turn 3dPP is NP-hard for any simple polygon $R'$ when using an axis-aligned unit square extruder.
\end{cor}
\begin{proof}
  $R$ in Theorem \ref{thm:mtpnphard} can be modified to a simple polygonal region $R'$ by adding narrow slits to connect holes in $R$ (Figure \ref{fig:Square_Block_First_Simply_Connected}).
  Let the width of each slit be $w/n'$ for $n' \geq n$ and $w \leq 1$.
  Addition of slits does not add any turn cost, since it is the same 3dPP with {\it total} idle movement $n(w/n') = \epsilon \leq 1$. 
  By Theorem \ref{thm:mtpnphard}, a Hamiltonian tour of length $n$ on $G$ gives a 3dPP tour with turn cost $6t_1 + 4t_2 $ for $R$. 
  This implies there is a 3dPP tour with turn cost $6t_1 + 4t_2$ on $R'$, since total idle movement is less than $1$.
  Conversely, let $R'$ have a tour with turn cost $6t_1 + 4t_2$.
  Then total turn cost of a tour in $R$ is also $6t_1 + 4t_2$, since idle movement in $R'$ is $< 1$. 
  Hence it corresponds to a Hamiltonian tour of length $n = t_1 + t_2$ in $G$ as implied by the arguments in Theorem \ref{thm:mtpnphard}.
\end{proof}

\section{\hspace*{-0.2in} SFCDecomp: \hspace*{-0.03in}Domain Decomposition and Space Filling Curves} \label{sec:SFCDecomp}

We develop \emph{SFCDecomp} as a framework that could handle large instances of 3dPP.
This framework uses a quadtree structure to decompose the integral orthogonal polygon (IOP) $P$ into square cells.
Some of the square cells can be joined to create bigger cells.
We then identify the traversal order of these cells using a Hilbert space filling curve.
The framework will work also with other rectilinear space filling curves such as Peano or Moore curves.

\subsection{Quadtree Decomposition and its Properties} \label{ssec:quadtree}
We create the quadtree decomposition of IOP $P$ as follows.
First, we find the smallest initial cell $[0, 2^q] \times [0, 2^q]$ that contains $P$ where $q>0$ is an integer.
Second, we apply adaptive subdivision of the initial cell until all cells are completely inside or outside of $P$.
We then remove cells that are completely outside of $P$ from the quadtree.
Third, we further subdivide cells if the area of the cell is more than $\delta > 0$, an integer.
An example of the quadtree is shown in Figure \ref{fig:quadtree_example_combine}.

We want to ensure that each leaf cell in the quadtree is traversed at least once by the toolpath algorithm.
One option is to do depth first traversal, where we recursively follow the left most unvisited branch until we reach a leaf cell.
Once done with the leaf cell, we backtrack until the first cell that has a left most unvisited child cell, and so on.
This gives us a sequential ordering of cells (see Example in Figure \ref{fig:quadtree_example_combinec}).
But jumps could be pretty large between neighboring cells in this sequence.
Instead, we create a sequence based on a Hilbert ordering such that neighboring cells are also adjacent (see Figure \ref{fig:quadtree_example_combined}). 

\begin{figure}[ht!]
  \centering
  \begin{subfigure}[t]{3.2in}
    \centering
    \includegraphics[scale=0.38]{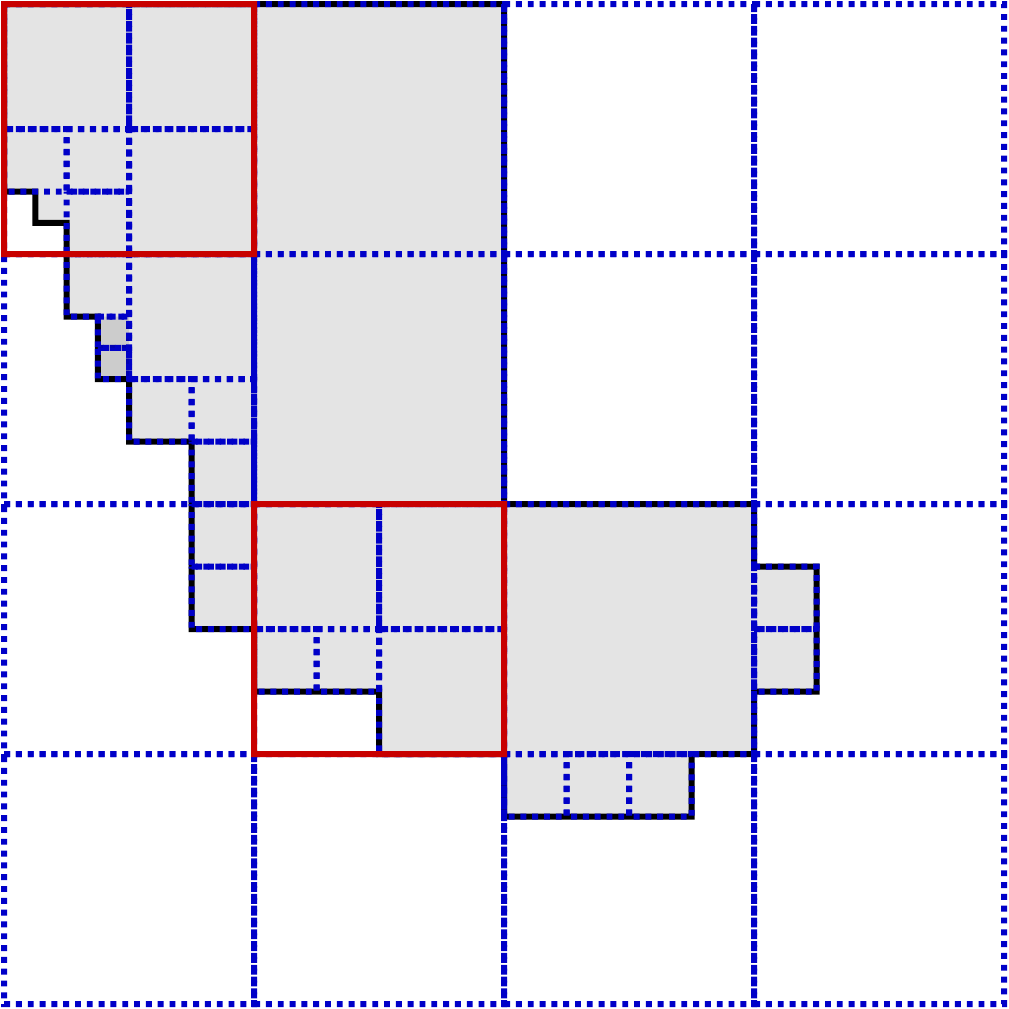}
    \caption{\label{fig:quadtree_example_combinea}}
    \medskip
  \end{subfigure}
  \begin{subfigure}[t]{3.2in}
    \centering
    \includegraphics[scale=0.38]{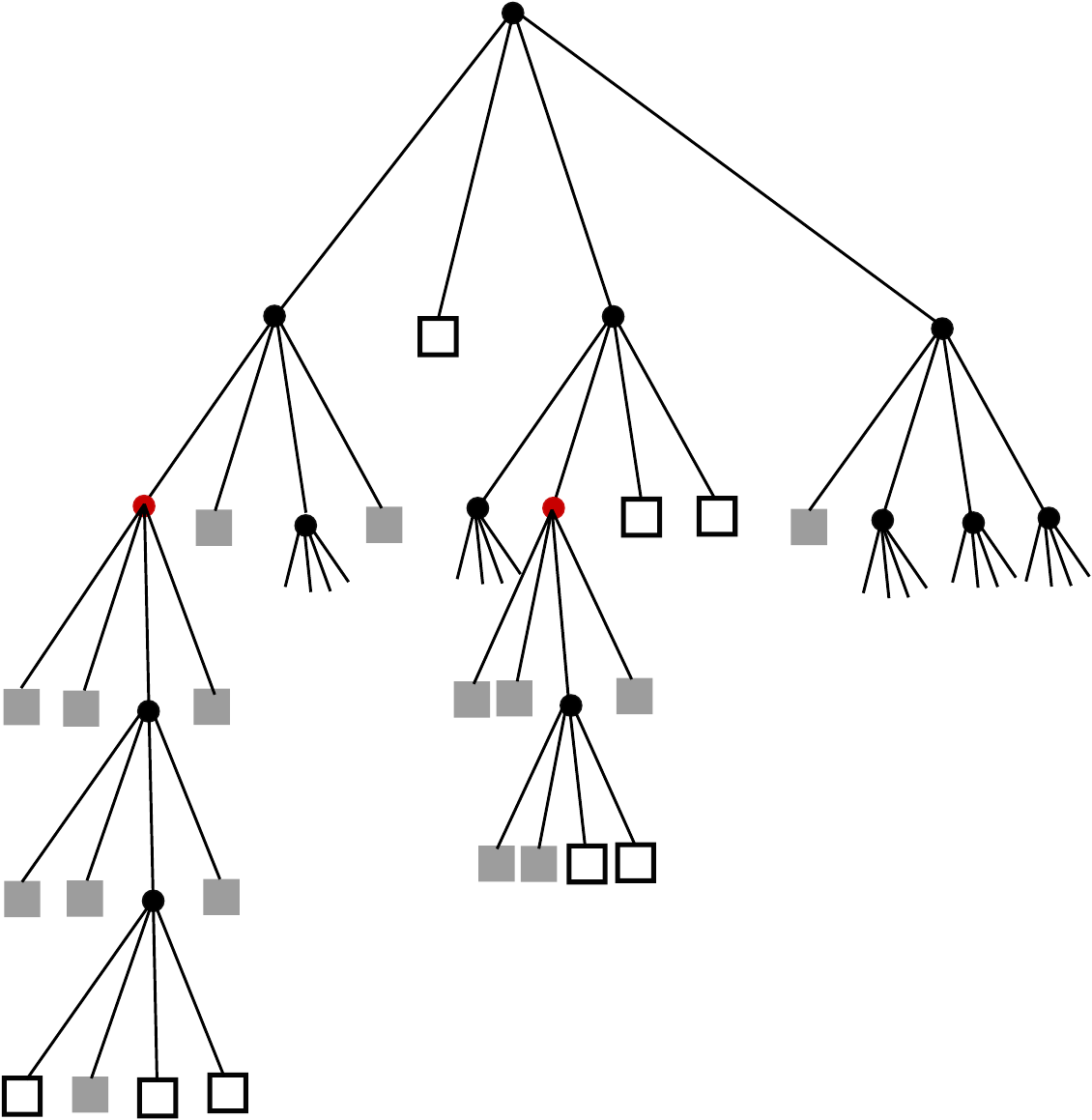}
    \caption{\label{fig:quadtree_example_combineb}}
    \medskip
  \end{subfigure}
  \\
  \bigskip
    \begin{subfigure}[t]{3.2in}
    \centering
    \includegraphics[scale=0.38]{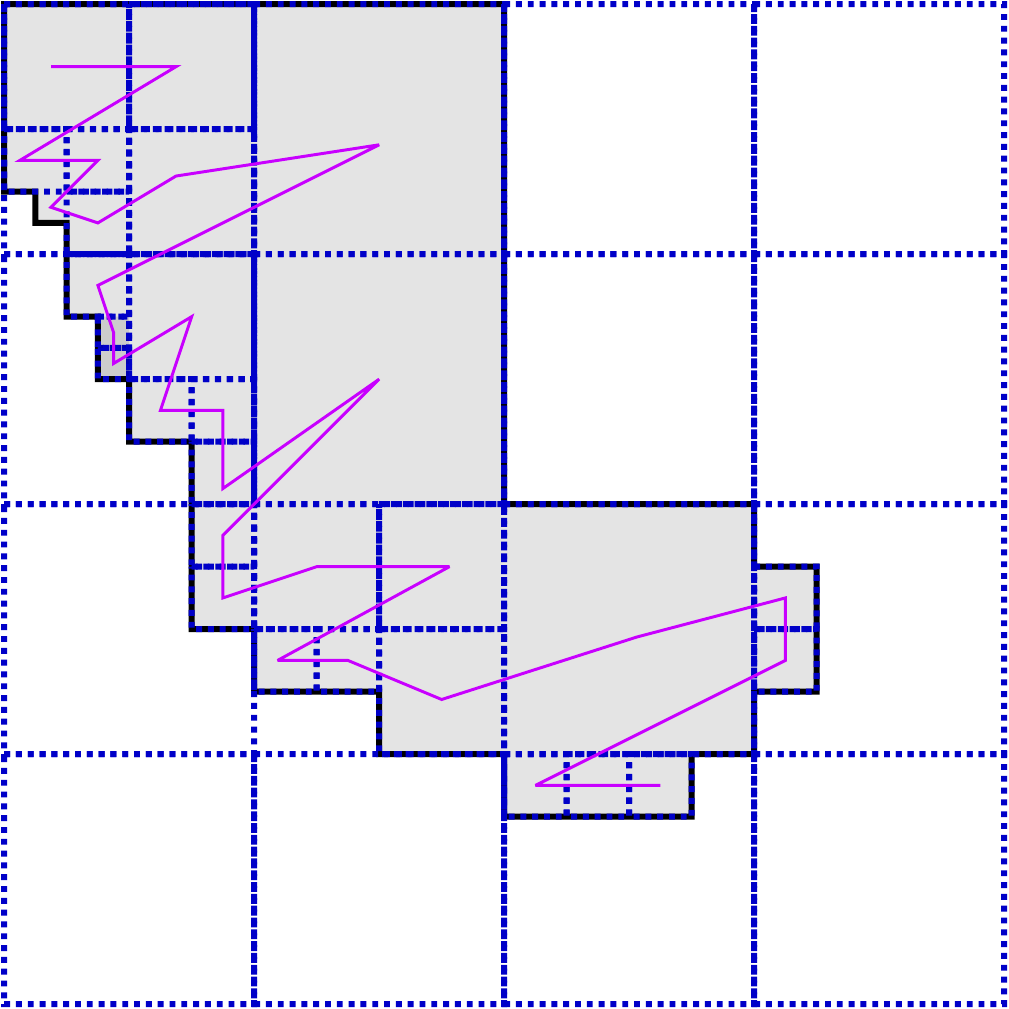}
    \caption{\label{fig:quadtree_example_combinec}}
  \end{subfigure}
  \begin{subfigure}[t]{3.2in}
    \centering
    \includegraphics[scale=0.38]{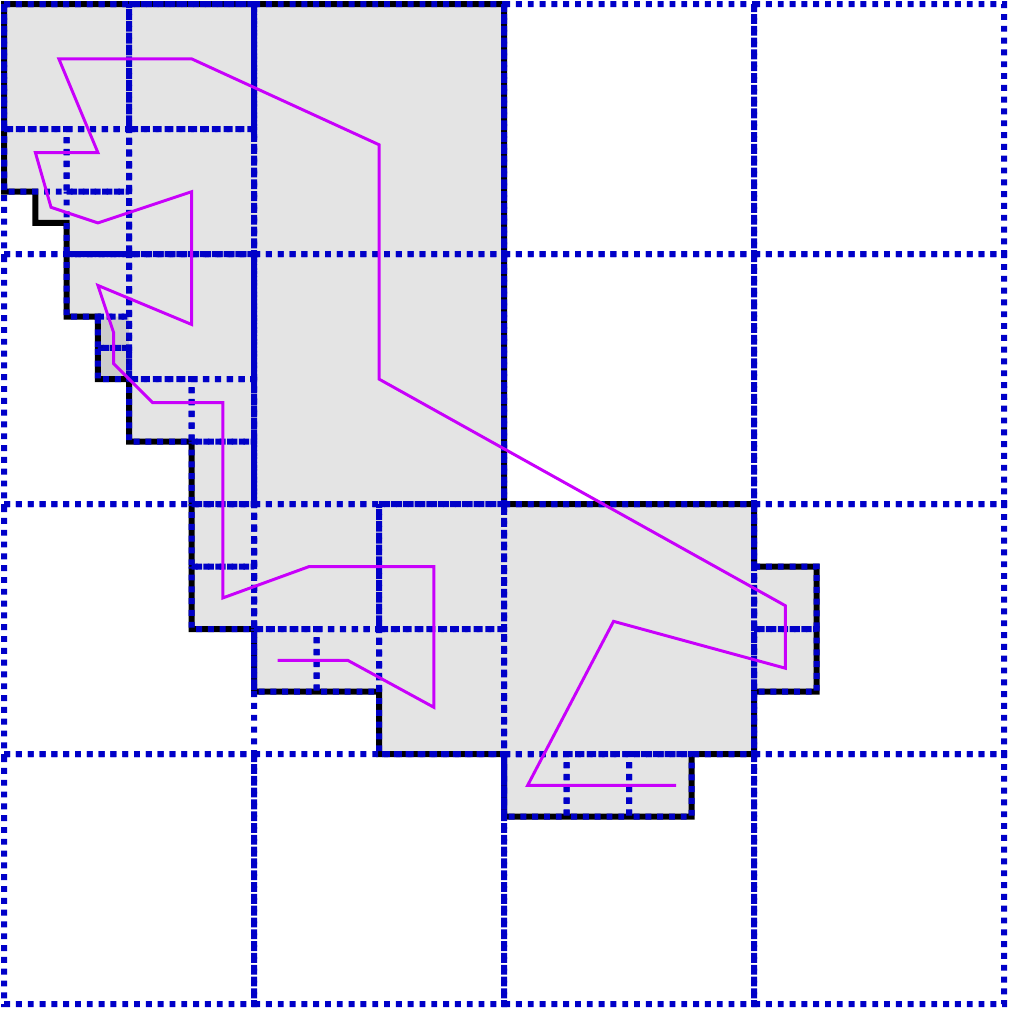}
    \caption{\label{fig:quadtree_example_combined}}
  \end{subfigure}
  \\
  \bigskip
  \caption{\label{fig:quadtree_example_combine}
    In Figure \ref{fig:quadtree_example_combinea}, red cells will be further subdivided since they are not completely inside the polygon.
    Figure (\ref{fig:quadtree_example_combineb}) shows the quadtree where red nodes represents red cells in Figure (\ref{fig:quadtree_example_combinea}) further subdivided, shaded square nodes represent cells completely inside polygon, and empty square nodes represent cells completely outside polygon and will be removed from the tree.
    Figures (\ref{fig:quadtree_example_combinec}) and (\ref{fig:quadtree_example_combined}) show sequential ordering of leaf node cells using depth first traversal and Hilbert ordering.
  }  	
\end{figure}


For illustration, consider an $8 \times 8$ square domain and its $2 \times 2$ cells generated from subdivision using quadtree with $\delta = 4$, and Hilbert ordering of its cells as shown in Figure \ref{fig:celltour}.
To generate a Hilbert curve we need to identify entry and exit corners for each cell.
Suppose the Hilbert curve enters at $(0, 0)$ and exits at $(8, 0)$.
Due to its properties, the Hilbert curve will enter and exit the cell at points that are along an edge of the cell (and not diagonally opposite), as shown in Figure \ref{fig:celltoura}.
Each cell $C$ (from now on, we refer to these cells as $C$ or $C_i$) in Figure \ref{fig:celltoura} is a union of pixels, as they are square with integer length.
Let $G$ be the dual graph of the pixel graph of cell $C$, and let entry ($s$) and exit ($t$) vertices in $G$ be the vertices closest to the Hilbert curve in $C$.
We ensure all the vertices in the dual graph of the pixel graph of cell $C$ are covered in the following way.
{\bfseries \textit{First}}, we find the path for each $G$ using the MIP model in Section \ref{sec:ipmodel} for a given choice of entry ($s$) and exit ($t$) vertices (Figure \ref{fig:celltourb}).
{\bfseries \textit{Second}}, we join these paths by a connecting path between exit and entry vertices of neighboring graphs in corresponding cells' Hilbert ordering (Figure \ref{fig:celltourc}).
If length of the connecting path is more than one unit, then it is set as idle movement of the extruder.
 
Note that the traversal of IOP can be a walk due to idle movements of the extruder.
We call a finite sequence of vertices and edges where both vertices and edges can be repeated a {\em walk}, and a {\it path} when vertices and edges are not repeated. 

\begin{figure*}[ht!] 
  \centering
  \begin{subfigure}[t]{2.1in}
    \includegraphics[scale=0.45]{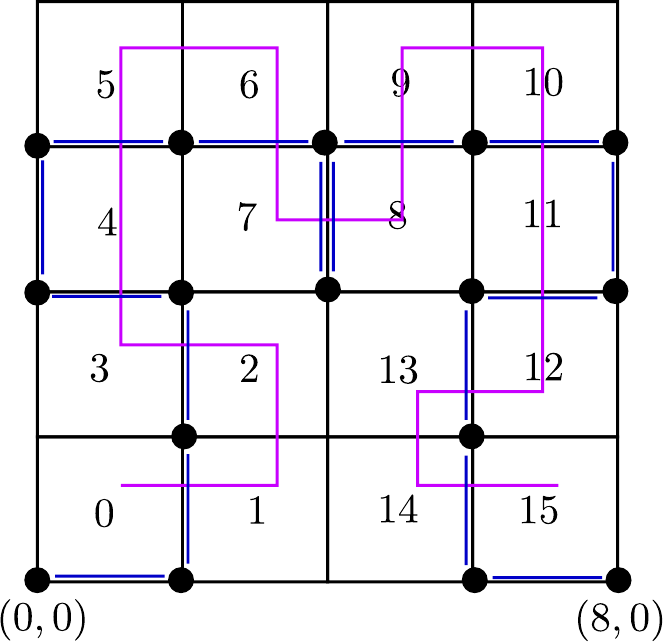}
    \caption{\label{fig:celltoura}}
  \end{subfigure}
  \begin{subfigure}[t]{2.1in}
    \centering
    \includegraphics[scale=0.45]{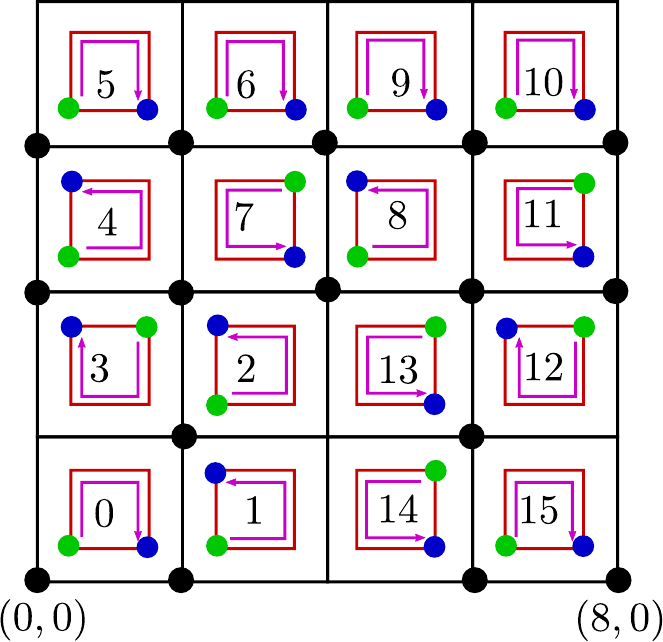}
    \caption{\label{fig:celltourb}}
  \end{subfigure}	
  \begin{subfigure}[t]{2.1in}
    \hspace*{0.1in}
    \includegraphics[scale=0.45]{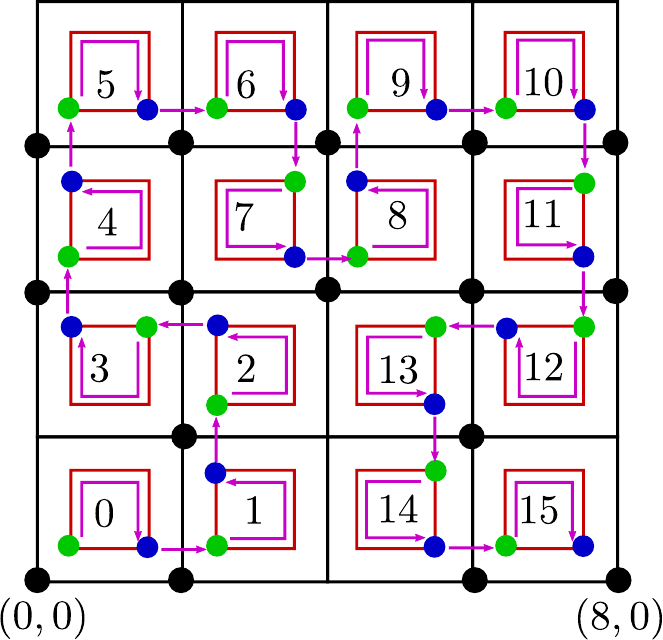}
    \caption{\label{fig:celltourc}}
  \end{subfigure}
  \caption{\label{fig:celltour} Black dots in Figure \ref{fig:celltoura} represent entry and exit vertices for each cell from $0$ to $15$, and Hilbert ordering of the cells is in pink.
    Figure \ref{fig:celltourb} shows entry vertex (green dot) and exit vertex (blue dot) for the dual graph (red) of the pixel graph of each cell.
    Paths shown in pink covers all vertices of dual graphs.
    Figure \ref{fig:celltourc} shows connecting paths in pink which connect paths on dual graphs based on Hilbert ordering of corresponding cells.} 
\end{figure*}

Entry and exit vertices of any $G$ identified by this approach are corner vertices in $G$ that can be joined by a straight path.
Let $(G, s, t)$ be a Hamiltonian $s$-$t$ path problem for  distinct vertices $s$ and $t$.
We show that $(G, s, t)$ always has a solution.
First we review certain properties of grid graphs.
A grid graph is a finite graph whose vertices are points with integer coordinates, and two vertices are connected if they are unit distance apart.
Let $(v_x, v_y)$ be the coordinates of vertex $v$.
Then $v$ is even if $v_x + v_y = 0 \mod 2$, else $v$ is odd.
This implies that grid graphs are bipartite, with edges connecting even and odd vertices.
Let $R(m ,n)$ be the grid graph whose vertex set is $\{v: 1 \leq v_x \leq m, ~ 1\leq v_y \leq n \}$.
A rectangular graph is a grid graph isomorphic to $R(m, n)$.
For the sake of completeness of presentation, we review the necessary conditions for existence of a Hamiltonian path in rectangular graphs presented by Itai et al.~\cite{ItPaSz1982}.
We also denote $R(m,n)$ by $B = (V^0 \cup V^1, E)$, the bipartite graph.
Since $B$ is two colorable, let all vertices in $V^0$ be of one color and all vertices in $V^1$ be of a second color.

\subsubsection{The Hamiltonian path problem $(B, s, t)$} \label{list:hamipathexist}
A solution to $(B, s,t)$ exists if at least one of the following conditions is {\it not} satisfied. 
\begin{enumerate}
  \item $B$ is even $(|V^0| = |V^1|)$ and $s$ and $t$ are of the same color, or $B$ is odd, say, with $|V^0| = |V^1| + 1$, and $s, t \in V^1$. 
  \item $n = 1$ and either $s$ or $t$ is not a corner vertex (Figure \ref{fig:hamiltonianconda}).
  \item $n = 2$ and edge $st$ is not a boundary edge (Figure \ref{fig:hamiltoniancondb}). 
  \item $n = 3$ and it satisfies following conditions (Figures \ref{fig:hamiltoniancondc}, \ref{fig:hamiltoniancondd}):
    \begin{enumerate}
      \item $s$ is different color from $t$, and $t$ is different color from the top corner vertices; and 
      \item $s_x \leq t_x - 1$ or $(s_y = 2$ and $s_x \leq t_x)$. 
    \end{enumerate}
\end{enumerate}

\begin{figure}[ht!] 
  \centering
  \begin{subfigure}[t]{3in}
    \centering
    \vspace*{-0.5in}
    \includegraphics[scale=0.47]{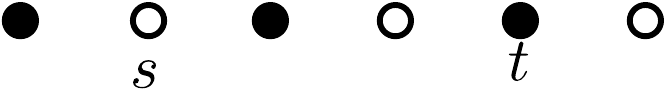}
    \vspace*{0.2in}
    \caption{\label{fig:hamiltonianconda}}
  \end{subfigure}
  \hfill
  \begin{subfigure}[t]{3in}
    \centering
    \includegraphics[scale=0.47]{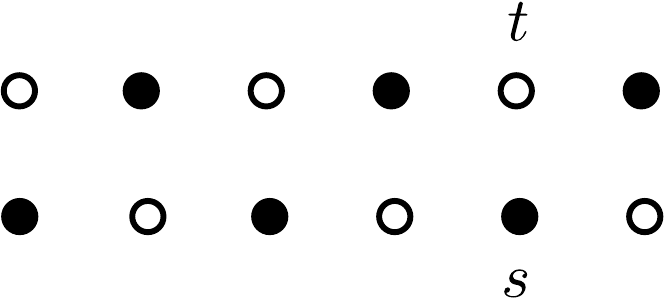}
    \caption{\label{fig:hamiltoniancondb}}
  \end{subfigure}
  \\
  \medskip
  \begin{subfigure}[t]{3in}
    \centering
    \vspace*{0.2in}
    \includegraphics[scale=0.47]{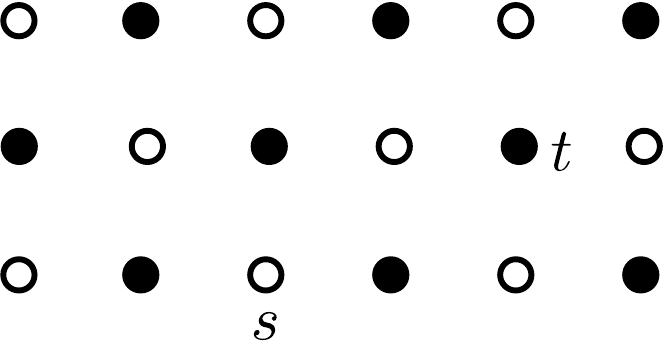}
    \caption{\label{fig:hamiltoniancondc}}
  \end{subfigure}
  \hfill
  \begin{subfigure}[t]{3in}
    \centering
    \vspace*{0.2in}
    \includegraphics[scale=0.47]{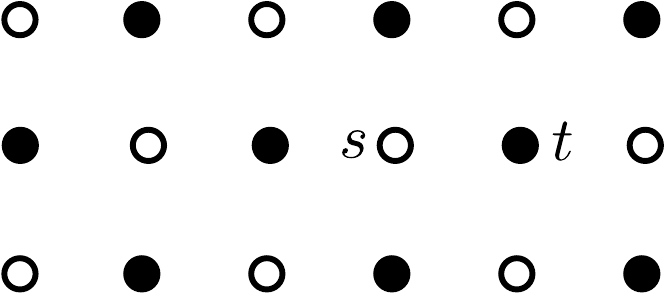}
    \vspace*{0.15in}
    \caption{\label{fig:hamiltoniancondd}}
  \end{subfigure}
  \caption{\label{fig:hamiltoniancond}
    Cases that prevent Hamiltonian path problem $(B,s,t)$ from having a solution.
    This figure is similar, but not identical, to Figure 3.1 in the work of Itai et al.~\cite{ItPaSz1982}.
  }
\end{figure}

\begin{lem}\label{lem:hamiltonianpathexist}
  Let $(G, s, t)$ be a Hamiltonian path problem on $G$, the dual graph of pixel graph of a cell $C$, and $s$ and $t$ be the entry and exit vertices of the Hamiltonian path.
  Then $(G, s, t)$ has a solution.
\end{lem}
\begin{proof}
  If area of $C$ is one unit then it is a trivial case, since $G$ has one vertex.
  More generally, the entry and exit vertices $s$ and $t$ are corner vertices in $G$, and are joined by a straight path.
  This straight path has even number of vertices.
  Hence $s$ and $t$ have different colors.
  Further, we have equal number of even and odd vertices in any $G$.
  Hence a Hamiltonian path always exists, as none of the conditions listed in Section \ref{list:hamipathexist} is satisfied.   
\end{proof}

\subsubsection{Joining square cells}\label{sssec:joincell}

We now try to join certain square cells into bigger cells so that we can solve larger instances of the subproblems to reduce turn costs.
Let $G_i$ be the dual graph corresponding to square cell $C_i$.
Let $S = \{C_1, \dots, C_h\}$ be the sequence of square cells based on Hilbert ordering for an IOP $P$, and $G(S) = \{G_1, \dots, G_h\}$ is the  corresponding sequence of dual graphs.
Let $(G_k, s_k, t_k)$ be the Hamiltonian path problems for $k \in [h] \coloneqq \{1, \dots, h\}$ and $s_k, t_k$ are chosen as in Section \ref{ssec:quadtree}.
Rectilinear distance between any two ordered neighbors $G_i, G_{i+1}$ in sequence $G(S)$ is defined as $d(G_i,G_{i + 1}) = |t_i^x - s_{i+1}^x| + |t_i^y - s_{i+1}^y|$.

We join square cells as follows.
{\bfseries{\textit{First}}}, let $G(S_i)= \{G_l, \dots, G_{l+k}\}$ where $d(G_{l-1},G_{l}) > 1$ or $l=1$, and $d(G_{l+k},G_{l+k+1})$ $>1$ or $l+k=h$ and $d(G_{i},G_{i+1})= 1 ~\forall i \in \{l, \dots, l+k-1\}$.
Find subsequence set $\{G(S_i)\}$ from $G(S)$.
{\bfseries{\textit{Second}}}, let $\tilde{C}_j$ be union of all the cells in subsequence $\tilde{S}_j$ of $S_i$.
Consider $\tilde{C}_j$ as an IOP.
Partition $S_i$ into a set of subsequences $\{\tilde{S}_j\}$ such that total area of $\tilde{C}_j$ is $\leq \Delta$, where $\Delta$ is maximum area allowed in any $\tilde{S}_j$ (see Figures \ref{fig:quadtreejoina}, \ref{fig:quadtreejoinb}).

\begin{lem}\label{lem:joinedgraphhamiltonianpathexist}
  Let $\{G_1, \dots, G_h\}$ be a sequence of dual graphs corresponding to sequence $\{C_1, \dots,$ $C_h\}$ whose union is $\tilde{C}_j$, $(G_k, s_k, t_k)$, $(\tilde{G}_i, s, t)$ be Hamiltonian path problems where $k \in [h]$ and $s_k, t_k$ $\forall k$ is chosen as described in Section \ref{ssec:quadtree}, and $\tilde{G}_i$ is the dual graph of the pixel graph of $\tilde{C}_i$.
  If $s= s_1$, $t = t_h$ then $(\tilde{G}_i, s, t)$ has a solution.
\end{lem}
\begin{proof}
  By Lemma \ref{lem:hamiltonianpathexist}, a Hamiltonian path exists for all $(G_i, s_i, t_i)$.
  The Hamiltonian path from $G_i$ is joined with the one from $G_{i+1}$ by an edge $\{t_i, s_{i+1}\}$ in $\tG_i$.
  The result follows once we set $s = s_1$, $t = t_m$.
\end{proof}

\begin{figure}[hb!]
  \medskip
  \centering
  \begin{subfigure}[t]{2.12in}
    \includegraphics[scale=0.42]{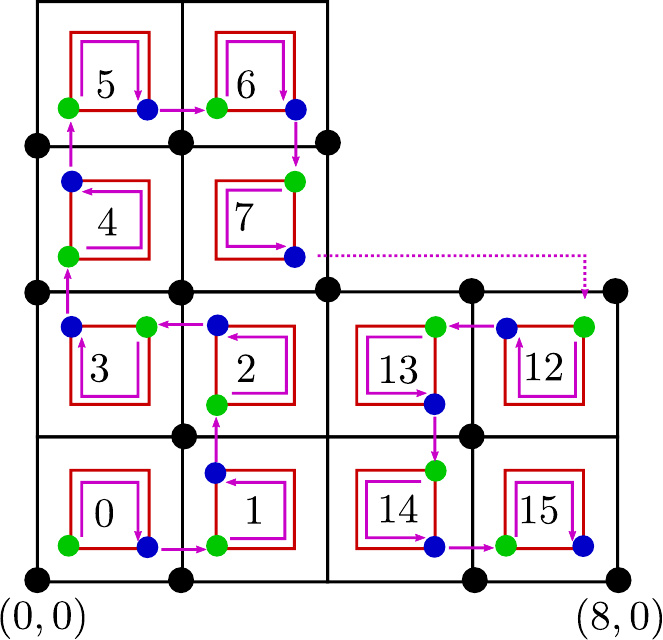}
    \caption{\label{fig:quadtreejoina}}
  \end{subfigure}
  \begin{subfigure}[t]{2.12in}
    \hspace*{0.07in}
    \includegraphics[scale=0.42]{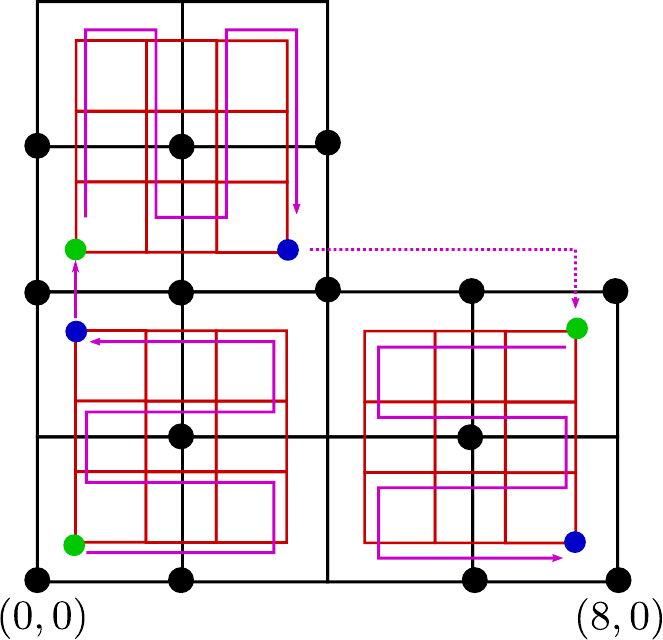}
    \caption{\label{fig:quadtreejoinb}}
  \end{subfigure}
  \begin{subfigure}[t]{2.12in}
    \hspace*{0.17in}
    \includegraphics[scale=0.42]{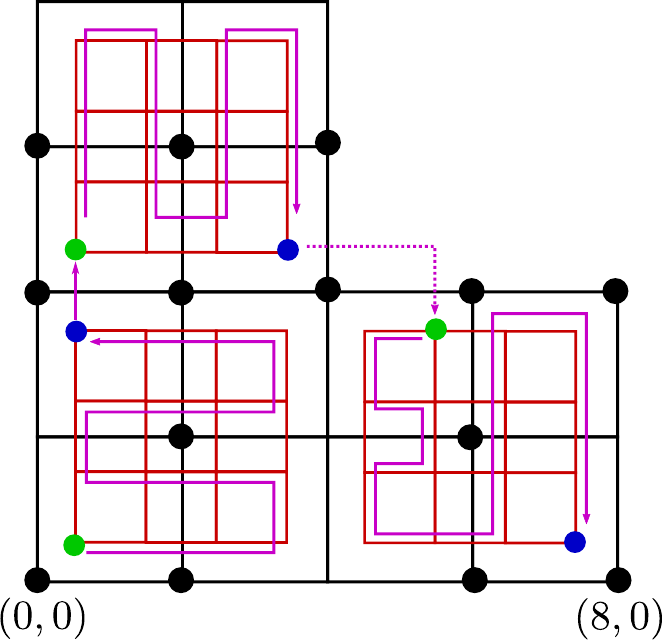}
    \caption{\label{fig:quadtreejoinc}}
  \end{subfigure}
  \caption{\label{fig:quadtreejoin}
    Path (pink) in Figure \ref{fig:quadtreejoina} covers the IOP.
    Idle movement of path is shown in pink dots.
    Further, Hilbert ordering of the cells is $[0, 1, 2, 3, 4, 5, 6, 7, 12, 13, 14, 15]$.
    Figure \ref{fig:quadtreejoinb} shows $3$ joined square cells based on Hilbert ordering into $[0, 1, 2, 3], [4, 5, 6, 7], [12, 13, 14, 15]$ with $\Delta = 16$ (Section \ref{sssec:joincell}).
    Figure \ref{fig:quadtreejoinc} shows updated entry vertex of the last joined square cell to reduce idle movement in the path (Section \ref{sssec:enterexitvertexupdate}).
  }  	
\end{figure}

\subsubsection{Update entry and exit vertices}\label{sssec:enterexitvertexupdate}

If vertices $v$ and $v'$ have both even or both odd coordinates, then we say that $v, v'$ have the same parity. 

\begin{lem}\label{lem:updatedjoinedgraphhamiltonianpathexist}
  Consider the same set up as in Lemma \ref{lem:joinedgraphhamiltonianpathexist}.
  If $s, t$ have same parity as $s_1, t_m$, respectively, and $s \in G_1$, $t \in G_m$, then $(\tilde{G}_i, s, t)$ has a solution.
\end{lem}
\begin{proof}
  By Lemma \ref{lem:joinedgraphhamiltonianpathexist}, $(\tilde{G}_i, s_1, t_m)$ has a solution.
  We know the pairs $\{s, t_1\}$ and $\{s_m, t\}$ have different parities.
  Then $(G_1, s, t_1)$ and $(G_m, s_m, t)$ have solutions since no condition in Section \ref{list:hamipathexist} is satisfied.
  Hence $(\tilde{G}_i, s, t)$ has a solution.  
\end{proof}

Let $\tilde{G}(S) = \{\tilde{G}_1, \dots, \tilde{G}_h\}$ be a sequence of dual graphs after joining square cells and let $(\tilde{G}_i, s_i, t_i)$ be the Hamiltonian path problems for each $\tilde{C}_i$.
Let subsequence $\tilde{S}_i$ have cells whose union is $\tilde{C}_j$ and $G(\tilde{S}_i) = \{G_1, \dots, G_k\}$.
We can reduce idle movement of the extruder by defining new Hamiltonian path problems $(\tilde{G}_i, s_{i}', t_{i}')$ for each $\tilde{C}_i$ such that $d(\tilde{G}_i, \tilde{G}_{i+1})$ is smallest compared to all possible choices of entry and exit vertices, where pairs $\{s_{i}, s_{i}'\}$ and $\{t_{i}, t_{i}'\}$ have same parity $\forall i$.
Further, we can have a case where $(\tilde{G}_i, s_{i}', t_{i}')$ has no solution if any condition in Section \ref{list:hamipathexist} is satisfied.
But based on Lemma \ref{lem:updatedjoinedgraphhamiltonianpathexist}, we can further add the restriction such that $s_{i}' \in G_1$ and $t_{i}' \in G_k$ to guarantee existence of a Hamiltonian path.
Note that the total idle movement in Figure \ref{fig:quadtreejoinb} is reduced in Figure \ref{fig:quadtreejoinc}. 

\textit{Correctness:} Based on our approach, Lemmas \ref{lem:hamiltonianpathexist}, \ref{lem:joinedgraphhamiltonianpathexist}, and \ref{lem:updatedjoinedgraphhamiltonianpathexist} guarantee existence of a Hamiltonian path in the dual graph of any cell even after implementing steps in Sections \ref{sssec:joincell} and \ref{sssec:enterexitvertexupdate}.   

\textit{Complexity:} Let $T$ be the maximum time to solve any problem $(\tilde{G}_i, s', t')$ and $N$ the total number of vertices in the dual graph of pixel graph of IOP.
Then time complexity is $O(NT) = O(N)$ when $T$ is small and fixed.
In practice, we observed $T$ in tens of seconds.

\section{Mixed Integer Programming Model and Relaxation}\label{sec:ipmodel}
Let $\tG$ be the dual graph for a cell in the decomposition of the region (similar to the illustration in Figure \ref{fig:inteortho3dprintprob}).
For a given pair of nodes $s,t$, we want to find a Hamiltonian $s$-$t$ path that minimizes a combination of total edge weights and turn costs.
Based on the Miller-Tucker-Zemlin (MTZ) formulation for TSP \cite{MiTuZe1960}, we present a mixed integer program (MIP) that also models turn costs.
With $x_{ij} \in \{0,1\}, i,j \in [n] \coloneqq \{1,\dots,n\}, i \neq j$, we let $x_{ij}=1$ if edge $(i,j)$ is included in the Hamiltonian path.
To avoid subtours, we let $u_i$ be the index of node $i \in [n]$ in the path (e.g., $u_i=4 \Rightarrow$ node $i$ is the $4$th node), and add the subtour elimination constraints in Equation (\ref{eq:IPsubtr}).
We let $c_{ij} \geq 0$ be the weight of edge $(i,j)$ that is part of the data, and use $c_i \geq 0$ to model the turn cost at node $i$.
To capture $c_i$, we set the binary parameter $A_{ijk} = 0$ if edges $(i,j)$ and $(j,k)$ form $180^\circ$ angle at $j$, and $A_{ijk}=1$ when this angle is $90^\circ$.
We then add the constraints in Equation (\ref{eq:IPturncost}).
The relative importance of edge and turn costs is captured by the scaling parameter $\alpha \in [0, 1]$, such that $\alpha=1$ gives the minimum edge cost problem and $\alpha=0$ gives the minimum turn cost problem.

\begin{align}
  \min_{x_{ij}}  \hspace*{0.1in} & \alpha \sum\nolimits_{(i,j) \in \tG} \, c_{ij} x_{ij} ~+~ (1-\alpha) \sum\nolimits_{i \in [n]} c_i \hspace*{0.1in}  \label{eq:IPobj} \\
  \begin{split}
    \mbox{s.t.} \hspace*{0.1in} & c_j ~\geq~A_{ijk}(x_{ij}+x_{jk}-1) \, \forall (i,j), (j,k) \in \tG;  \\
                & c_j ~\geq~0~ \forall j \in [n];~~ c_s = c_t = 0;
  \end{split} \label{eq:IPturncost}
  \\
  \begin{split}
    & \sum\nolimits_{j} x_{ij} = \sum\nolimits_{j} x_{ji} \geq 1, i \in [n], i \neq s,t; \\
    & \sum\nolimits_{j} x_{sj} = 1,\sum\nolimits_{j} x_{js} = 0;
      \sum\nolimits_{j} x_{tj} = 0,\sum\nolimits_{j} x_{jt} = 1;
  \end{split} \label{eq:IPnodeblnce}
  \\
  & u_i - u_j + 1 \leq n(1-x_{ij})~\forall (i,j) \in \tG, i \notin{s,t};~~ u_s = 1; \label{eq:IPsubtr} \\
  & x_{ij} \in \{0,1\} ~\forall (i,j)  \in \tG. \nonumber
\end{align}

To solve relatively large instances of this IP efficiently, we remove the subtour constraints in Equation (\ref{eq:IPsubtr}) and develop the following heuristic.
\begin{enumerate}
  \item\label{itm:joincycles}
    {\it {\bfseries Join Cycles:}}
    \begin{enumerate}
     \item \label{step:iprelax} Solve the relaxed MIP without constraints (\ref{eq:IPsubtr})) for $(\tilde{G}_i, s, t)$.
       Generally, We obtain a set of cycles and an $s$-$t$ path covering all the vertices in $\tG_i$.
       
     \item Create a new undirected graph $G_i^c$ where each vertex represents a cycle obtained in Step (\ref{step:iprelax}) above, and there is an edge between vertices if the corresponding cycles can be joined into a bigger cycle by performing a $2$-opt exchange \cite{ApBiChCo2007} at a unit square between them.
       If there are multiple such unit squares, we pick one that adds the minimum weight to the joined cycle.
       Solve a {\it minimum spanning forest} (MSF) problem on $G_i^c$, and join cycles based on the edges in each tree in the MSF.
       Repeat this step until the solution does not change.
    \end{enumerate}

  \item \label{itm:joincyclesandpath}
    {\it {\bfseries Join Cycles and $s$-$t$ Path:}}
    Examine the cycles created by Step (\ref{itm:joincycles}) in the increasing order of numbers of squares available for $2$-opt exchange with the $s$-$t$ path.
    Update the $s$-$t$ path by merging cycles in this order, choosing the minimum cost square for $2$-opt exchanges in each step.
\end{enumerate}

Note that this heuristic is not guaranteed to identify a Hamiltonian $s$-$t$ path for every $(\tG_i, s, t)$.
But in over 10,000 such instances over all layers of the Buddha and the Bunny models in our computations, it failed to identify a Hamiltonian path in only \emph{one} instance.

\textit{Complexity:} The heuristic runs in $O((n^c_i)^2 T_{\mbox{\scriptsize IP}}) = O(n^2 T_{\mbox{\scriptsize IP}})$ time, where $n^c_i = O(n)$ is the number of nodes in $G_i^c$ and $T_{\mbox{\scriptsize IP}}$ is the time for solving the relaxed IP instance.
  Even though we relax the problem by removing constraints (\ref{eq:IPsubtr}), it is still an MIP and hence $T_{\mbox{\scriptsize IP}}$ could vary exponentially with problem size and data in the worst case \cite{S86}.
  While most coefficients in the constraint matrix of the relaxed IP are $0$ or $1$, the $A_{ijk}$ terms in constraint (\ref{eq:IPturncost}) could take any nonnegative value depending on the problem instance.
But in practice, $T_{\mbox{\scriptsize IP}}$ was usually in tens of seconds and the MSF computation ran in $O(n^c_i \log n^c_i)$ time for $n^c_i \ll n$.
  In a more general SFCDecomp framework, one could consider using an approximation algorithm, e.g., a modification of the Christofides--Serdyukov algorithm \cite{Ch1976,Se1978,vBeSl2020}, to identify a Hamiltonian cycle for each subdomain rather than solve the MIP to optimality.
  Another option may be to consider standard approaches such as zigzag patterns to design the toolpath within the subdomain.
  In such cases, $T_{\mbox{\scriptsize IP}}$ will vary polynomially with problem size.

\subsection{General Geometry}\label{ssec:gengeometryextension}

We describe how to extend the SFCDecomp framework to handle general geometries that are not integral orthogonal polygons.
Let $\mathbb{V}$ be the set of vertices of print edges in the tool path (Section \ref{sec:SFCDecomp}) closest to the boundary of general polygon $\mathbb{P}$ containing the IOP $P$ (i.e., $P \subset \mathbb{P}$ in the nontrivial case).
Then we can project vertices in $\mathbb{V}$ to the boundary of $\mathbb{P}$.
We project each vertex in $\mathbb{V}$ \emph{at most once} (see Figure \ref{fig:gengeomtryext}).

\begin{figure}[ht!] 
  \centering
  \includegraphics[width=4.5in]{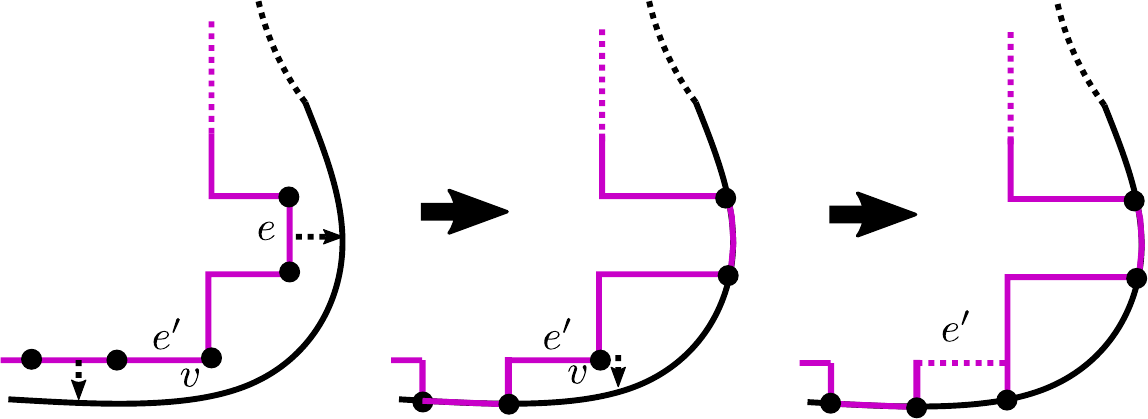}
  \caption{\label{fig:gengeomtryext}
    We project edge $e$ with vertices in $\mathbb{V}$ orthogonal to $e$ onto the boundary of $\mathbb{P}$.
    Project all such edges.
    Some vertices in $\mathbb{V}$ may not get projected.
    Suppose $v \in \mathbb{V}$  in print edge $e'$ is not projected (middle).
    Then convert $e'$ to idle movement in the tool path, and project $v$ to the boundary of $\mathbb{P}$ orthogonal to $e'$ (right).
    Here, edge $e'$ could be idle to start with.
  }  	
\end{figure}

Results from an implementation of the complete pipeline with the MIP model (Section \ref{sec:ipmodel}) for a single layer is shown in Figure \ref{fig:singlebunnylayer}.
We can handle disconnected polygons, or ones with holes.
Figure \ref{fig:introimage} (in Page \pageref{fig:introimage}) and Figure \ref{fig:bunnybuddha} show sample layers from implementations on the Buddha and the Bunny models. 

\begin{figure}[ht!] 
  \begin{subfigure}[t]{2.1in}
    \includegraphics[scale=0.45]{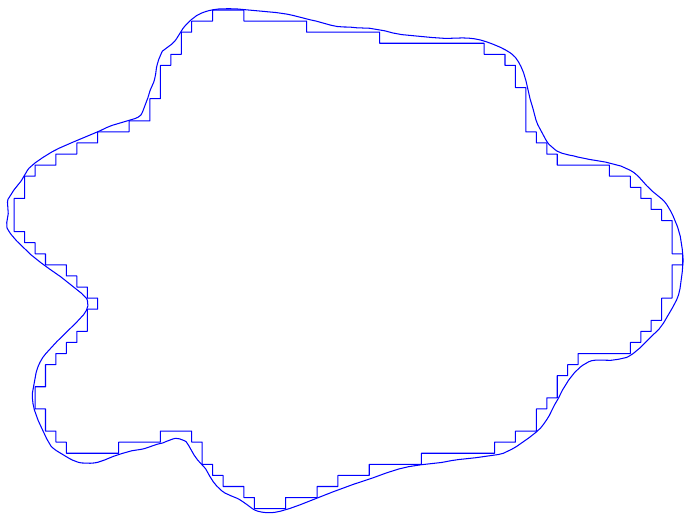}
    \caption{\label{fig:singlebunnylayera}}
  \end{subfigure}
  \begin{subfigure}[t]{2.1in}
    \includegraphics[scale=0.45]{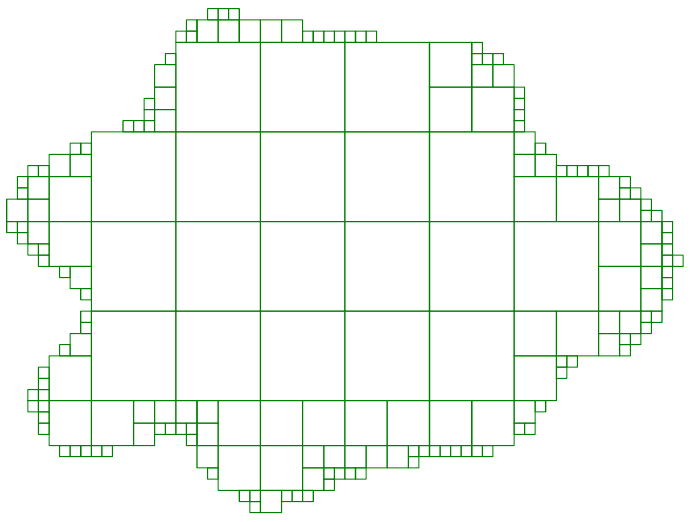}
    \caption{\label{fig:singlebunnylayerb}}
  \end{subfigure}
  \begin{subfigure}[t]{2.1in}
    \includegraphics[scale=0.45]{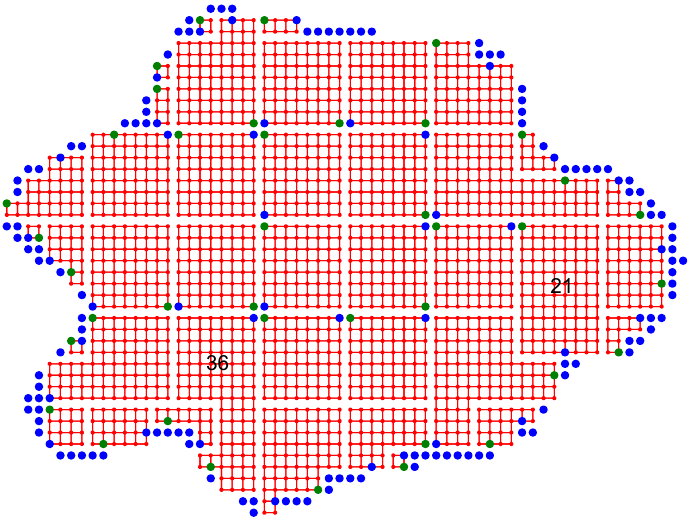}
    \caption{\label{fig:singlebunnylayerc}}
  \end{subfigure}
  \\
  \begin{subfigure}[t]{2.1in}
    \includegraphics[scale=0.45]{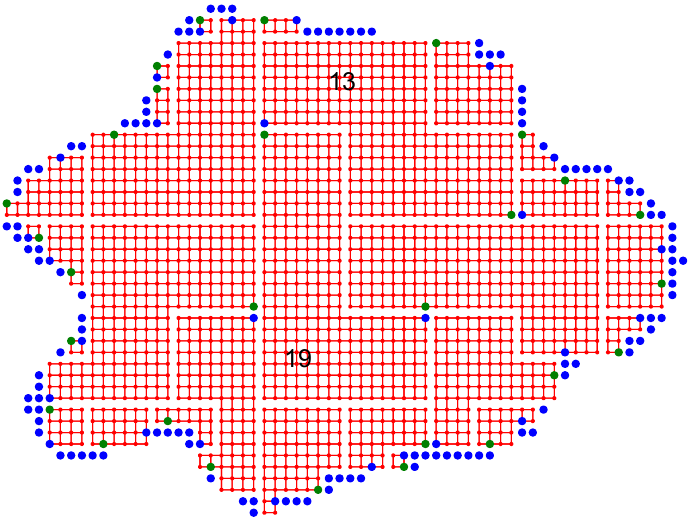}
    \caption{\label{fig:singlebunnylayerd}}
  \end{subfigure}
  \begin{subfigure}[t]{2.1in}
    \includegraphics[scale=0.45]{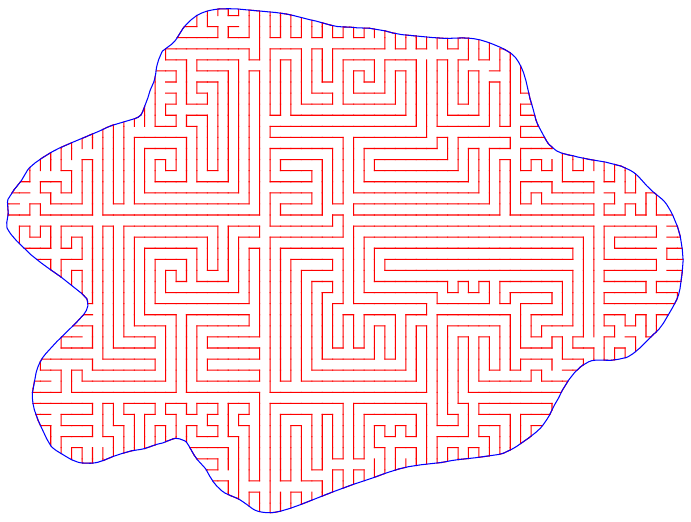}
    \caption{\label{fig:singlebunnylayere}}
  \end{subfigure}
  \begin{subfigure}[t]{2.1in}
    \includegraphics[scale=0.45]{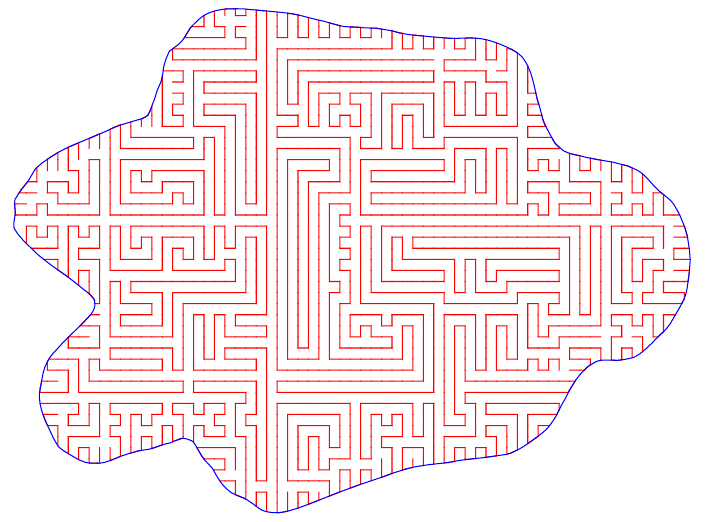}
    \caption{\label{fig:singlebunnylayerf}}
  \end{subfigure}	
  \caption{\label{fig:singlebunnylayer}
    Figure \ref{fig:singlebunnylayera} shows polygon $\mathbb{P}$ and its largest IOP $P$ in blue.
    Figure \ref{fig:singlebunnylayerb} shows decomposition of $P$ with $\delta = 64$.
    Figures \ref{fig:singlebunnylayerc} and \ref{fig:singlebunnylayerd} show dual graphs (red) after implementing Sections \ref{sssec:joincell}, \ref{sssec:enterexitvertexupdate} with $\Delta=120, 256$ (resp.).
    Enter and exit vertices are shown in green and blue dots (also isolated vertices).
    Print paths (red) with uniform and random edge weights are in Figures \ref{fig:singlebunnylayere} and \ref{fig:singlebunnylayerf} (both combined with turn costs).
    } 
\end{figure}

\section{Maximizing or Minimizing Edge Overlaps} \label{sec:maxminovlp}

As a direct application of the SFCDecomp framework, we consider maximizing or minimizing print edge overlaps across adjacent layers. 
We consider the overlap of two edges from adjacent layers is maximum when they coincide, i.e., one edge is printed entirely over the other.
The edge overlap is minimum when the two edges intersect at most in a single point.
See Figure \ref{fig:mechtest}c for an illustration.
The premise we want to investigate is that the extent of edge overlap across adjacent layers affects the mechanical strength of the printed object.

To maximize (resp.~minimize) overlap, the weight of all edges in the dual graph $\tilde{G}_i$ of Layer $i$ is set to $0.5$ (resp.~$1.5$) instead of $1$ (default) if they overlap with edges in the print path in Layer $(i-1)$.
Figure \ref{fig:maxminoverlapturnplot} shows variations of edge overlap ratios across the initial $148$ layers of the Stanford Bunny. 
We make the following observations.

\begin{enumerate}
  \item \label{itm:maxminoverlapa}
    {\bfseries Figure \ref{fig:maxminoverlapturnplota}:}
    Original (uniform edge weights) and maximum overlap problems have similar overlap costs.
    For same $(\tilde{G}_i, s, t)$ between adjacent layers, both maximum overlap and original problems should return same solution.
    We can observe large changes in edge overlap at a few places in the original problem.
    It is due to significant variation in decompositions between adjacent layers.
    For instance, layers $33, 34, 35$ have significant differences in decomposition as shown in Figures \ref{fig:maxorivariationandminvariationa}, \ref{fig:maxorivariationandminvariationb}, \ref{fig:maxorivariationandminvariationc}.
    Further, the desired general trend of increase in overlap for maximum overlap and original problems, as well as decrease in overlap for minimum overlap, are observed due to increase in cross section area of the layers as we move up in $z$-axis.    

  \item\label{itm:maxminoverlapb}
    {\bf Figure \ref{fig:maxminoverlapturnplotb}:}
    In general there are sharp changes in number of turns across any $3$ adjacent layers in the minimum overlap problem since it forces adjacent layers to have different tours.
    It further forces Layer $i+2$ to have a similar toolpath as Layer $i$.
    This is illustrated for the same Block $5$ in layers $117, 118, 119$ in Figures \ref{fig:maxorivariationandminvariatione}, \ref{fig:maxorivariationandminvariationf}, \ref{fig:maxorivariationandminvariationg}. 
\end{enumerate}

\begin{figure*}[ht!] 
  \centering
  \begin{subfigure}[t]{3.2in}
    \centering
    \includegraphics[scale=0.60]{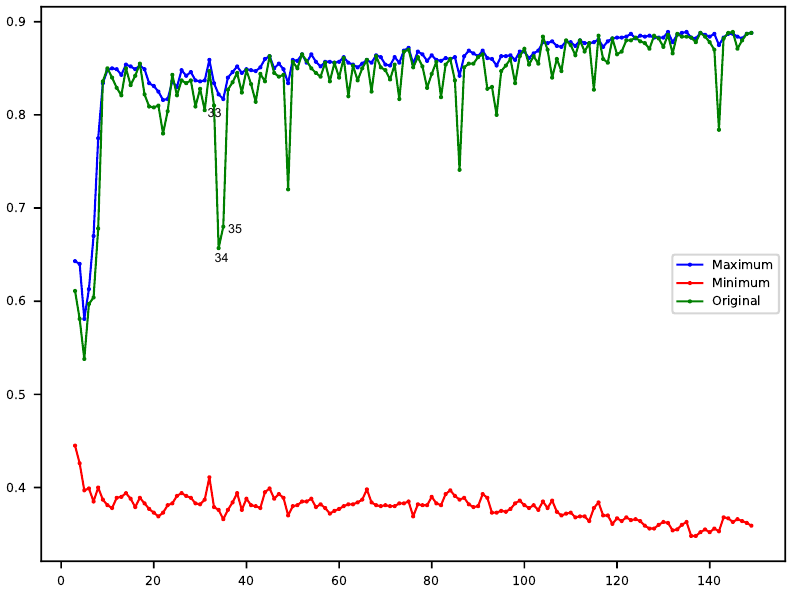}
    \caption{\label{fig:maxminoverlapturnplota}}
  \end{subfigure}
  \begin{subfigure}[t]{3.2in}
    \centering
    \includegraphics[scale=0.60]{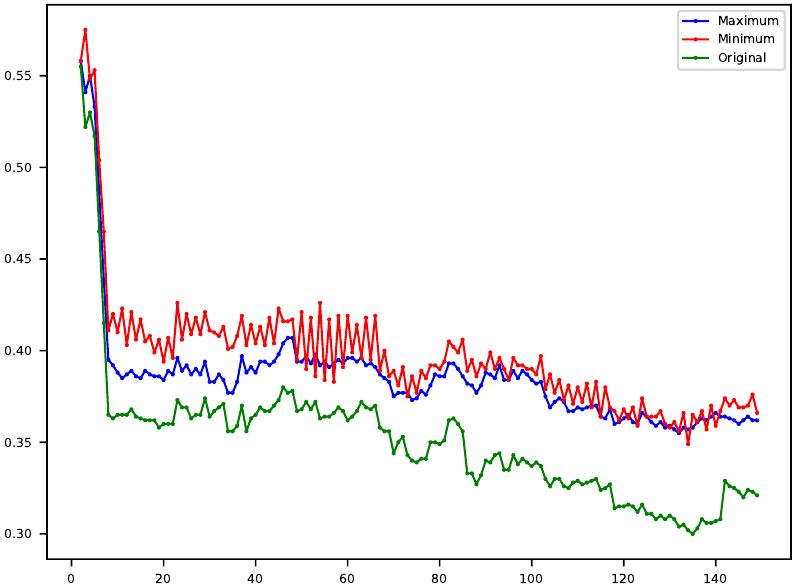}
    \caption{\label{fig:maxminoverlapturnplotb}}
  \end{subfigure}
  \caption{\label{fig:maxminoverlapturnplot}
    Figure \ref{fig:maxminoverlapturnplota} shows the Ratio (Total Overlapped printed edges / Total printed edges) for each Layer and Figure \ref{fig:maxminoverlapturnplotb} shows the Ratio (Total 90 degree Turns / Total 90 and 180 degrees Turns) for each Layer for maximization, minimization, and using uniform edge weights (original) problems on the initial 148 layers of the Bunny model.
    Note that the horizontal axes in both plots list the layers numbered from 1--148.
  }
\end{figure*}

\begin{figure*}[hb!] 
  \centering
  \begin{subfigure}[t]{2.12in}
    \includegraphics[scale=0.45]{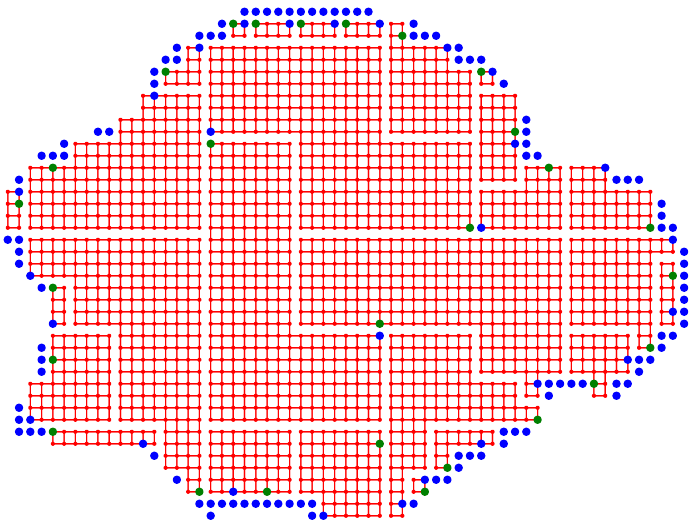}
    \caption{\label{fig:maxorivariationandminvariationa}}
  \end{subfigure}
  \begin{subfigure}[t]{2.12in}
    \centering
    \includegraphics[scale=0.45]{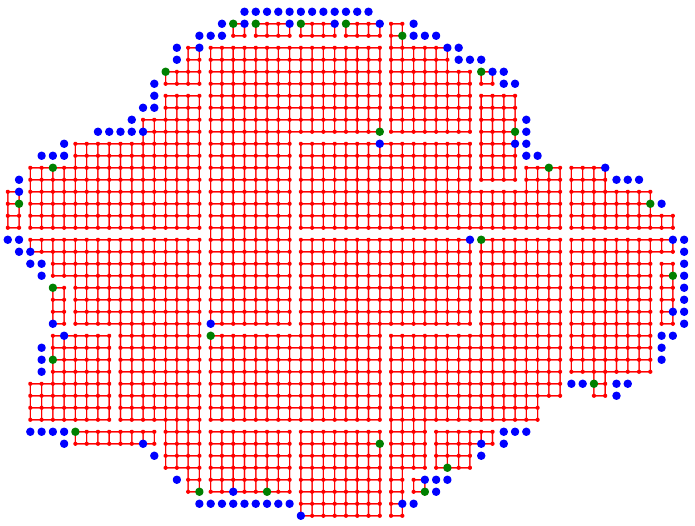}
    \caption{\label{fig:maxorivariationandminvariationb}}
  \end{subfigure}
  \begin{subfigure}[t]{2.12in}
    \includegraphics[scale=0.45]{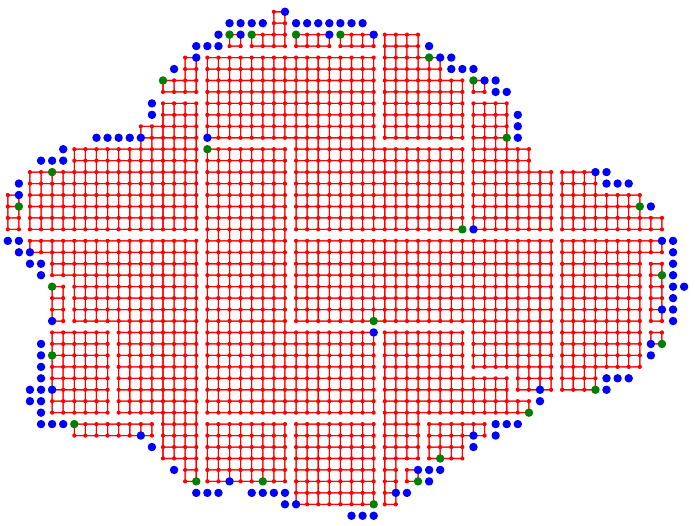}
    \caption{\label{fig:maxorivariationandminvariationc}}
  \end{subfigure}
  \begin{subfigure}[t]{2.12in}
    \centering
    \includegraphics[scale=0.45]{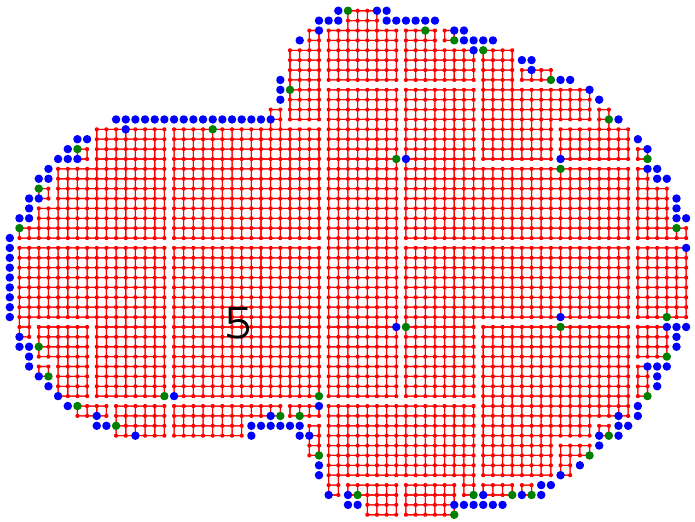}
    \caption{\label{fig:maxorivariationandminvariationd}}
  \end{subfigure}
  \begin{subfigure}[t]{1.4in}
    \centering
    \includegraphics[scale=0.65]{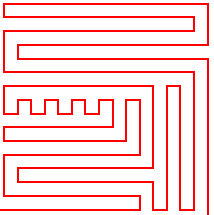}
    \caption{\label{fig:maxorivariationandminvariatione}}
  \end{subfigure}
  \begin{subfigure}[t]{1.4in}
    \quad
    \includegraphics[scale=0.65]{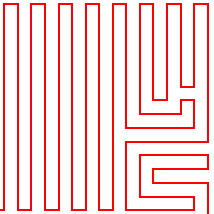}
    \caption{\label{fig:maxorivariationandminvariationf}}
  \end{subfigure}
  \begin{subfigure}[t]{1.4in}
    \quad
    \includegraphics[scale=0.65]{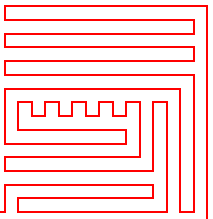}
    \caption{\label{fig:maxorivariationandminvariationg}}
  \end{subfigure}
  \caption{\label{fig:maxorivariationandminvariation}
    Figures \ref{fig:maxorivariationandminvariationa}, \ref{fig:maxorivariationandminvariationb}, \ref{fig:maxorivariationandminvariationc} show decompositions for layers 33, 34, and 35 of the Bunny, respectively.
    Figure \ref{fig:maxorivariationandminvariationd} shows dual graph of layer $117$ with Block $5$ (present in layers $117, 118, 119$).
    Figures \ref{fig:maxorivariationandminvariatione}, \ref{fig:maxorivariationandminvariationf}, \ref{fig:maxorivariationandminvariationg} show tool path in those layers with minimum overlap on Block $5$.
  }
\end{figure*}

\begin{figure*}[ht!] 
  \centering
  \includegraphics[width=3.2in]{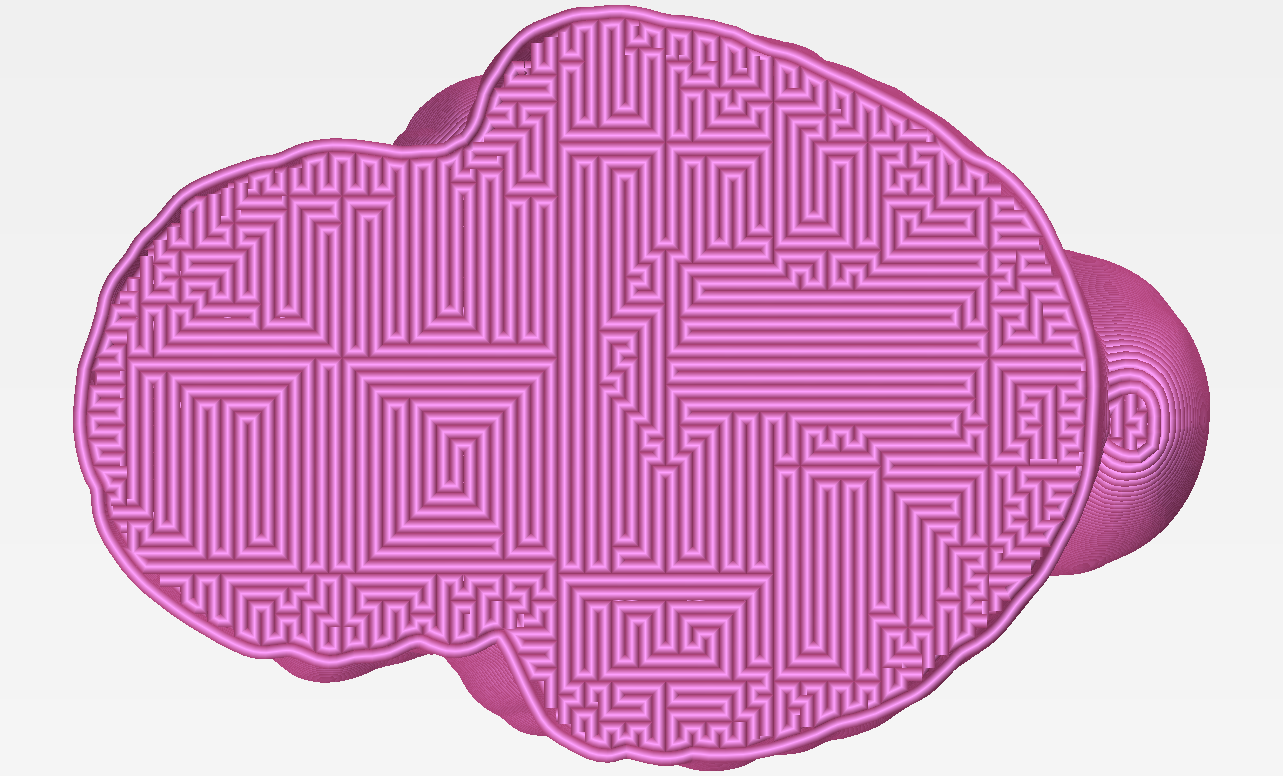}
  \includegraphics[width=3.2in]{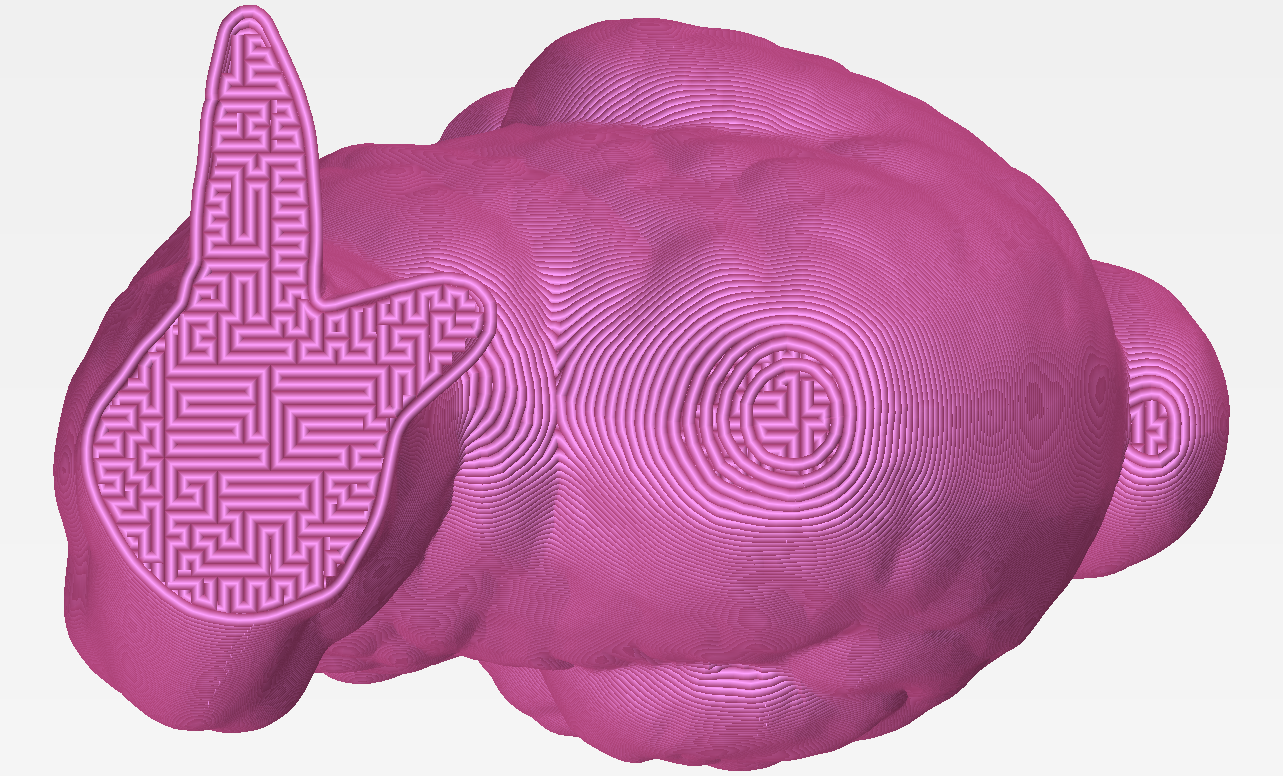}\\
  \includegraphics[width=3.2in]{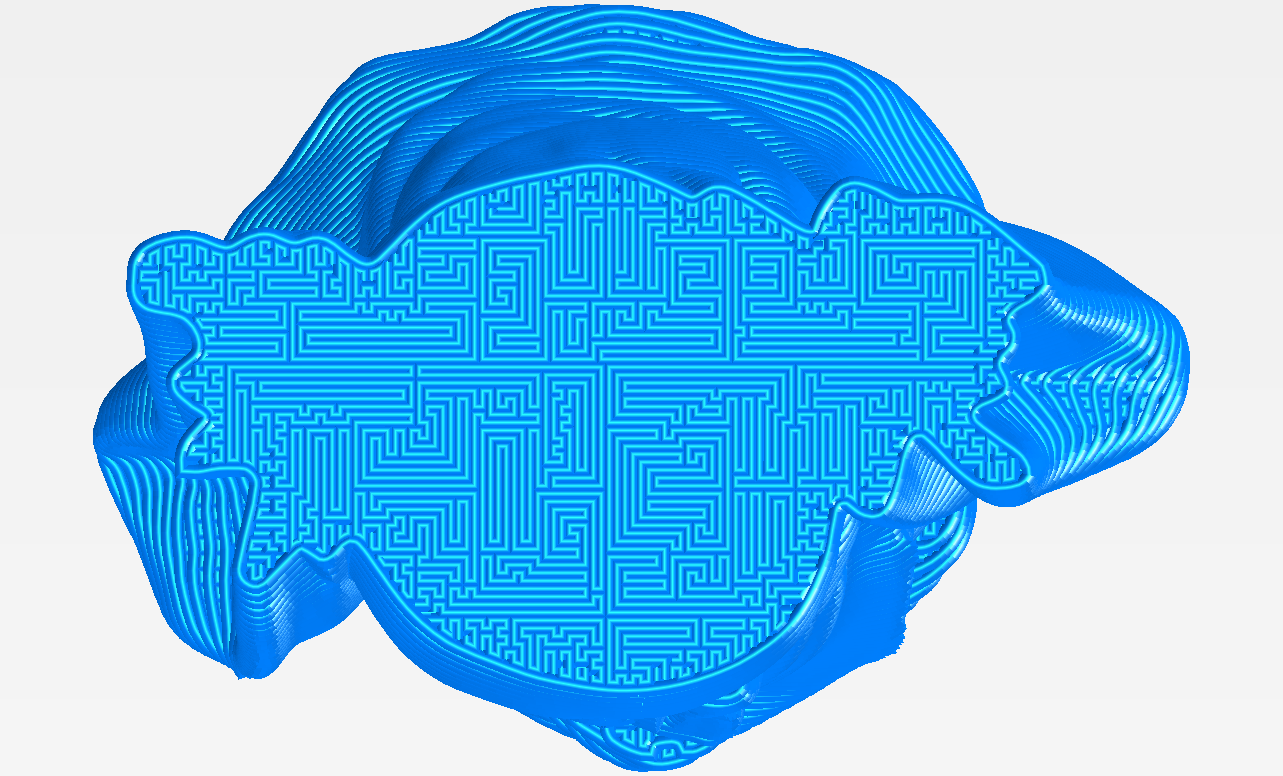}
  \includegraphics[width=3.2in]{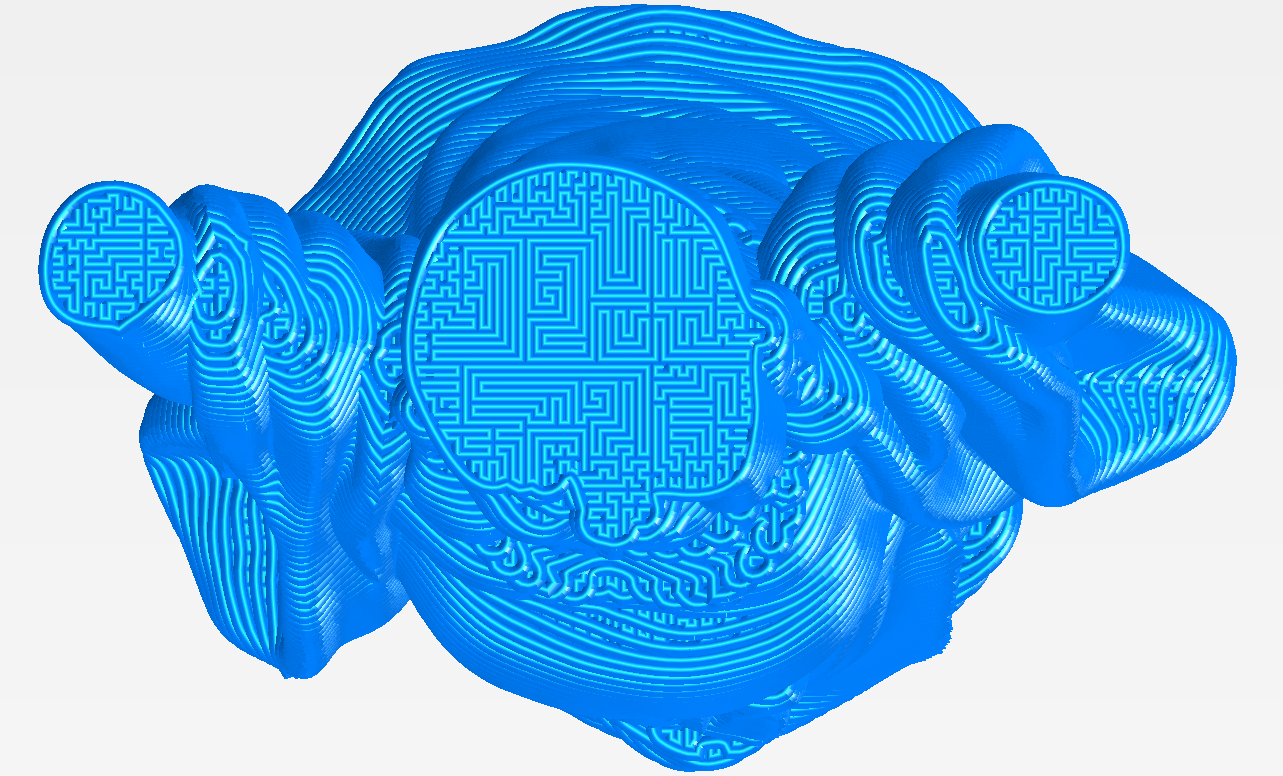}  
  \caption{\label{fig:bunnybuddha} Layers 125 and 312 of Bunny model (top) and layers 90 and 138 of Buddha model (bottom).}
\end{figure*}

\section{Mechanical Testing} \label{sec:mechtest}

We experimentally investigated the effects of using two space filling curve designs generated by SFCDecomp that provide maximum (Max) and minimum (Min) interlayer edge overlap (Figure \ref{fig:mechtest}c).
  We hypothesized that the contact surface area between adjacent layers is larger in a Max pattern as compared to a Min pattern.
  Therefore a Max pattern should have better interlayer bonding and strength along the layer stacking direction.
We found that specimens with minimum and maximum edge overlaps have significantly different tensile behavior.

\subsection{Experimental Design and Fabrication}

We created tensile specimens shaped as a rectangular cuboid with a hole in the center (Figure \ref{fig:mechtest}a).
The layers were stacked along the $z$-axis and the tensile test was conducted with the load applied along the $x$-axis (Figure \ref{fig:mechtest}c).
Ideally, the tensile load should be applied along the layer stacking direction to directly investigate the effect of different interlayer edge overlap patterns on interlayer adhesion.
But due to the limitation in equipment and material, a specimen with a large enough geometry to include enough subdomains was not feasible.
Hence a flat rectangular cuboid specimen was implemented.
For uniform placement on the print platform, the first layer for all samples was printed using the same zigzag design at 0.2 mm height.
The specimens had 19 layers at 0.2 mm height on the base layer.
We used a polylactic acid (PLA) filament with 1.75 mm diameter in a Cartesian desktop fused filament fabrication (FFF) device with a nozzle of 0.4 mm in diameter (from MatterHackers Inc., USA).
We measured tensile mechanical properties using a universal tensile testing machine (INSTRON 600DX) along with an extensometer (Epsilon 3542-0200-50-ST) of 2'' gauge length and a load cell rated for 1000 lbs.
We performed the tests with a constant strain rate of 0.8 mm/minute.

Printing many orthogonal line segments with 90$^\circ$ turns (with extrusion width of 0.4 mm) caused several holes to appear at intersections of multiple extrusion corners.
Hence we extruded 10\% more material near a 90$^\circ$ turn so as to reduce the size and number of these holes.
We fabricated the specimens at a hot-end temperature of 200$^\circ$C, print-bed temperature of 45$^\circ$C, and extrusion speed of 30 mm/s and acceleration limit of 500 m/s$^2$.
We tested six samples each with maximum (Max) and minimum (Min) edge overlaps across adjacent layers.
We repeated two distinct pairs of print patterns for adjacent layers that satisfied the criteria for maximum and minimum edge overlaps.

\begin{figure}[ht!]
  \includegraphics[width=\textwidth]{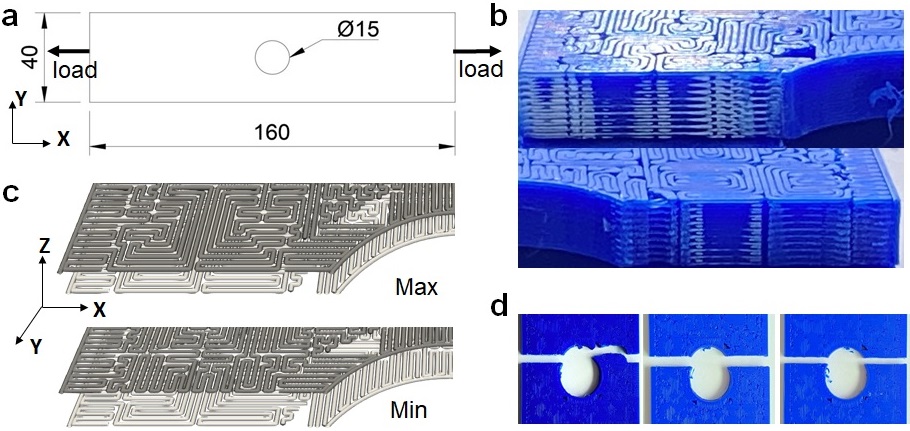}
  \caption{\label{fig:mechtest}
    Mechanical testing.
    a) Dimensions of specimen (in mm) with a hole in the middle.
    b) Cross sections of failure.
    c) Sample adjacent layers for maximum (``Max'' at top) and minimum (``Min'' at bottom) edge overlap designs.
    d) Failure lines were similar across most specimens.
  }
\end{figure}

\subsection{Results and Discussion}

Results from the tensile tests are shown in Table \ref{tab:mechtestResults} and the stress--strain curves shown in Figure \ref{fig:StrStrCrvs}.
The fracture surfaces suggested the specimens failed primarily due to localized delamination \cite{DoJu2019} between the extrusion filaments.
These failures are along sections where a large number of extrusions were printed in the direction perpendicular to the load direction (Figure \ref{fig:mechtest}b) and often coincide with cell boundaries.
Within each set of the Max and Min samples, the fractures occurred mostly at the same site (Figure \ref{fig:mechtest}d).

\begin{table}[ht!]
  \caption{  \label{tab:mechtestResults}
    Test results of specimens with maximum edge overlap (Table \ref{tab:MaxResults}, left) and with minimum edge overlap (Table \ref{tab:MinResults}, right).
    Mdls gives the tensile modulus in GPa, Strgth gives the tensile strength in MPa, and Elngtn gives the elongation at break (as \%).
    The last two rows in each table list the averages (Max/Min\_Mean) and standard deviations (Max/Min\_Stdev) of the measurements.
  }
  \centering
  \begin{subtable}[t]{0.42\textwidth}
    \caption{\label{tab:MaxResults} Maximum overlap samples.}
    \begin{tabular}{|l|c|c|c|} \hline
      Sample &  Mdls  & Strgth & Elngtn  \\ \hline \hline
      Max\_1 &  2.67 &  21.67 & 1.03 \\ \hline
      Max\_2 &  2.48 &  13.43 & 0.70 \\ \hline
      Max\_3 &  2.63 &  12.22 & 0.57 \\ \hline
      Max\_4 &  3.16 &  17.39 & 0.74 \\ \hline
      Max\_5 &  2.56 &  13.24 & 0.70 \\ \hline
      Max\_6 &  2.55 &  12.93 & 0.79 \\ \hline \hline
   Max\_Mean &  2.68 &  15.15 & 0.76 \\ \hline 
  Max\_Stdev &  0.25 &   3.68 & 0.15 \\ \hline
    \end{tabular}
  \end{subtable}
  \hspace*{0.4in}
  \begin{subtable}[t]{0.42\textwidth}
    \caption{\label{tab:MinResults} Minimum overlap samples.}
    \begin{tabular}{|l|c|c|c|} \hline
      Sample &  Mdls  & Strgth & Elngtn  \\ \hline \hline
      Min\_1 &  2.53 &  17.46 & 0.76 \\ \hline
      Min\_2 &  2.42 &  15.14 & 0.69 \\ \hline
      Min\_3 &  2.60 &  17.94 & 0.80 \\ \hline
      Min\_4 &  2.54 &  17.45 & 0.76 \\ \hline
      Min\_5 &  2.63 &  18.63 & 0.81 \\ \hline
      Min\_6 &  2.73 &  18.95 & 0.79 \\ \hline \hline
      Min\_Mean &  2.57 &  17.60 & 0.77 \\ \hline
      Min\_Stdev &  0.10 &   1.35 & 0.04 \\ \hline 
    \end{tabular}
  \end{subtable}
\end{table}

The first Max sample (Max\_1) exhibited significantly higher tensile strength and elongation at break compared to other Max specimens (Table \ref{tab:MaxResults} and left plot in Figure \ref{fig:StrStrCrvs}).
This outlier may be due to manufacturing variation that resulted in overall lower porosity in this sample, thus improving the tensile properties.
Removing this outlier decreases further the tensile strength of the Max samples (mean: 13.84, standard deviation: 2.04).
All Min samples exhibited similar tensile behavior and properties (Table \ref{tab:MinResults} and right plot in Figure \ref{fig:StrStrCrvs}).
Overall, the minimum edge overlap samples exhibited better tensile strength (mean: 17.60, standard deviation: 1.35) and modulus than the maximum edge overlap ones while the elongation at break is similar for both sets.

The stress-strain curves (Figure \ref{fig:StrStrCrvs}) reveal similarity in tensile moduli of the maximum and minimum overlap patterns.
All specimens experienced elongation at break and tensile strength less than the material specification.
The reduced elasticity could be due to the combined effect of intralayer adhesion arising from extrusion segment orientation \cite{ToRo2016}, and the type of dominant defects \cite{FaMoKa2019}  present as the result of different interlayer pattern generation strategy.
 Although the two pattern generating strategies maximized and minimized interlayer edge overlap, the effect of interlayer adhesion was not investigated in detail due to the tensile load direction being orthogonal to the layer stacking direction. 

\begin{figure}[ht!]
  \centering
  \includegraphics[width=3.2in]{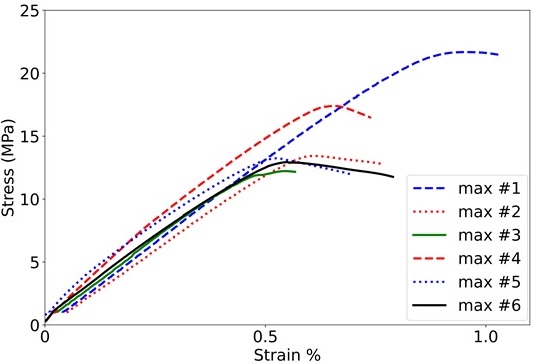}
  \hspace*{0.02in}
  \includegraphics[width=3.2in]{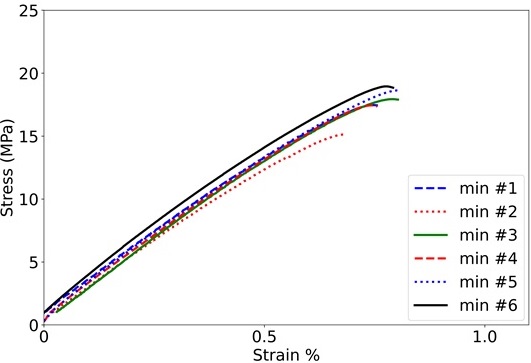}
  \caption{\label{fig:StrStrCrvs}
    Stress-Strain curves for maximum edge overlap specimens (left) and minimum edge overlap specimens (right).
  }
\end{figure}

 In a typical Cartesian-based additive manufacturing device, the toolhead motion is decomposed into components along the $x,y$, and $z$ axes and are actuated independently in each axis \cite{YoLi2016}.
 Any turn incurs a time cost of the actuator having to accelerate to overcome the momentum of the toolhead to change direction.
 On the other hand, many of the conventional infill patterns such as contour-parallel and spiral suffer from directional bias.
 In applications that purposely require a slow print speed such as geometrically accurate hydrogel bioprinting \cite{WeDo2017} or thermoplastic with enhanced mechanical properties \cite{ChChVe2016}, the actuators could accelerate to the target speed within a very short amount of time and employing SFCDecomp may be particularly beneficial.
 Other cases that prioritize directional uniformity more than turn cost may also utilize SFCDecomp by choosing parameters appropriately in the framework.

\section{Discussion}
\vspace*{-0.1in}

We have developed a decomposition approach to solve large instances of optimized path planning problem in 3D printing where each sub-polygon is guaranteed to have a dual graph with a feasible tool path. 
Our framework guarantees that discontinuities in the tool path, if any, are located only at the boundary of the original (input) polygon.
Further, we can change the Hilbert ordering of the cells by changing the enter and exit corner vertices of the initial cell in the quadtree.
We can also create various decompositions for the same IOP by changing the values of parameters $\delta$ and $\Delta$.

The edge weights in our graph framework could model multiple quality factors including turn costs, edge overlap across adjacent layers, tool path length, and others.
Our mechanical testing has shown that changing the extent of overlap across layers could impact the mechanical strength of the printed object.
Another potential application of our framework is the optimization of internal microstructure and thermal management by choosing appropriately defined edge weights derived from physical models and/or experiments, which in turn could result in increased strength of the printed objects.

For the Buddha and Bunny models together, our framework solved more than 10,000 MIP subproblems across more than 500 layers.
These MIP instances could be solved independently, and hence in an embarrassingly parallel fashion.
Alternatively, our framework allows the use of approximation algorithms or heuristics to solve the subproblems instead of solving the MIP model to optimality \cite{ApBiChCo2007}.
We could also reuse optimal solutions for cells that reoccur across multiple layers.
Using uniform edge weights, and varying the relative importance parameter $\alpha$ (Equation \ref{eq:IPobj}), we could obtain fractal-like patterns for the toolpath.  

If we solve the full MIP model including subtour constraints (\ref{eq:IPsubtr}), any discontinuities in the tool path will be located at the boundary of the original polygon. 
This setting could be of concern when the polygon is relatively thin as compared to size of extruder.
We could consider reducing the extruder size to handle such situations, or consider alternative methods (e.g., spiral or zigzag patterns).
In extreme cases where the polygon in a given layer has many sharp curvature regions, we could have many small sized sub-polygons near the boundary.
This setting could create several discontinuities in the tool path at the boundary.
We will explore methods to handle such extreme cases in our future work.

Our mechanical testing experiments (Section \ref{sec:mechtest}), while preliminary, already demonstrate that the amount of edge overlap across adjacent layers could significantly affect the strength of the print.
We plan to employ the SFCDecomp framework to study in detail the effects of tool path design as well as edge overlap across layers on various mechanical properties.

\medskip
\noindent {\bfseries Acknowledgment:}
Gupta and Krishnamoorthy acknowledge partial funding from the US National Science Foundation through grants 1661348 and 1819229.

\vspace*{-0.2in}



\end{document}